%% file: main.tex
\documentclass[a4paper,UKenglish,cleveref, autoref, thm-restate]{lipics-v2021}

\usepackage{graphicx}
\usepackage{microtype}
\usepackage[linesnumbered,ruled,vlined]{algorithm2e}
\usepackage{url}
\usepackage{wrapfig}

\title{Optimal-Time Queries on BWT-runs Compressed Indexes} 

\titlerunning{Optimal-Time Queries on BWT-runs Compressed Indexes} 

\author{Takaaki Nishimoto}{RIKEN Center for Advanced Intelligence Project, Japan}{takaaki.nishimoto@riken.jp}{}{}
\author{Yasuo Tabei}{RIKEN Center for Advanced Intelligence Project, Japan}{yasuo.tabei@riken.jp}{}{}
\authorrunning{T. Nishimoto and Y. Tabei}
\Copyright{T. Nishimoto and Y. Tabei}


\category{} 

\relatedversion{} 

\nolinenumbers 

\EventEditors{John Q. Open and Joan R. Access}
\EventNoEds{2}
\EventLongTitle{42nd Conference on Very Important Topics (CVIT 2016)}
\EventShortTitle{CVIT 2016}
\EventAcronym{CVIT}
\EventYear{2016}
\EventDate{December 24--27, 2016}
\EventLocation{Little Whinging, United Kingdom}
\EventLogo{}
\SeriesVolume{42}
\ArticleNo{23}

\begin{document}

\maketitle

\input{abstract}
\input{macros}
\input{1_intro.tex}

\input{2_preliminaries.tex}
\input{3_block_permutation}

\input{4_backward_search}

\input{4_r-index-f}

\input{5_applications}

%
%
%

\bibliographystyle{plainurl}
\bibliography{ref}
\clearpage
\input{9_appendix}

\end{document}

%% file: abstract.tex
\begin{abstract}
Indexing highly repetitive strings (i.e., strings with many repetitions) for fast queries 
has become a central research topic in string processing, because it has a wide variety 
of applications in bioinformatics and natural language processing.
Although a substantial number of indexes for highly repetitive strings have been proposed thus far, 
developing compressed indexes that support various queries remains a challenge. 
The \emph{run-length Burrows-Wheeler transform} (RLBWT) is a lossless data compression by 
a reversible permutation of an input string and run-length encoding, 
and it has received interest for indexing highly repetitive strings.
LF and $\phi^{-1}$ are two key functions for building indexes on RLBWT, and 
the best previous result computes LF and $\phi^{-1}$ in $O(\log \log n)$ time 
with $O(r)$ words of space for the string length $n$ and the number $r$ of runs in RLBWT. 
In this paper, we improve LF and $\phi^{-1}$ so that they can be computed in a constant time with $O(r)$ words of space. 
Subsequently, we present \emph{OptBWTR (optimal-time queries on BWT-runs compressed indexes)}, 
the first string index that supports various queries including locate, count, 
extract queries in optimal time and $O(r)$ words of space. 
\end{abstract}

%% file: macros.tex
\newcommand{\Occ}{\mathit{Occ}}

\newcommand{\floor}[1]{\left \lfloor #1 \right \rfloor}
\newcommand{\ceil}[1]{\left \lceil #1 \right \rceil}

\newcommand{\argmax}{\mathop{\rm arg~max}\limits}
\newcommand{\argmin}{\mathop{\rm arg~min}\limits}
\newcommand{\polylog}{\mathop{\rm polylog}\limits}

\newcommand{\SA}{\mathsf{SA}}

\newcommand{\LF}{\mathsf{LF}}
\newcommand{\FL}{\mathsf{FL}}
\newcommand{\select}{\mathsf{select}}
\newcommand{\rank}{\mathsf{rank}}

\newcommand{\rlbwt}{\mathsf{rlbwt}}

\newcommand{\pred}{\mathsf{pred}}

\newcommand{\Linterval}{\mathsf{Linterval}}

\newcommand{\splitf}{\mathsf{split}}
\newcommand{\balance}{B}
\newcommand{\phibalance}{\dot{B}}

\newcommand{\movef}{\mathsf{Move}}

\newcommand{\sbwt}{L_{\mathsf{first}}}

\newcommand{\flbwt}{L_{\mathsf{FL}}}


%
%

%% file: 1_intro.tex
\section{Introduction}\label{sec:introduction}
A string index represents a string in a compressed format that supports locate queries (i.e., computing all the positions at which a given pattern appears in a string). 
The \emph{FM-index} ~\cite{DBLP:journals/jacm/FerraginaM05,DBLP:journals/talg/FerraginaMMN07} is an efficient string index 
on a lossless data compression called the \emph{Burrows-Wheeler transform (BWT)}~\cite{burrows1994block}, which is a reversible permutation of an input string.
In particular, locate queries can be efficiently computed on an FM-index by performing a \emph{backward search}, which is an iterative algorithm for computing an interval 
corresponding to the query on a \emph{suffix array} (SA)~\cite{DBLP:journals/siamcomp/ManberM93} storing all the suffixes of an input string in lexicographical order.
The FM-index performs locate queries in $O(m + occ)$ time 
with $O(n (1 + \log \sigma / \log n))$ words of space  for a string $T$ of length $n$, query string of length $m$, 
alphabet size $\sigma$, and number $occ$ of occurrences of a query in $T$.

A \emph{highly repetitive string} is a string including many repetitions. 
Examples include the human genome, version-controlled documents, and source code in repositories. 
A significant number of string indexes on various compressed formats for highly repetitive strings have been proposed thus far~(e.g., SLP-index~\cite{DBLP:conf/spire/ClaudeN12a}, LZ-indexes~\cite{DBLP:conf/latin/ChristiansenE18,DBLP:conf/latin/GagieGKNP14}, 
BT-indexes~\cite{DBLP:journals/talg/ChristiansenEKN21,DBLP:journals/tcs/NavarroP19}). 
For a large collection of highly repetitive strings, the most powerful and efficient compressed format is 
the run-length (RL) Burrows Wheeler transform (RLBWT)~\cite{burrows1994block}, which is a BWT compressed by run-length encoding. 
M\"{a}kinen et al.~\cite{DBLP:journals/jcb/MakinenNSV10} presented 
an RLBWT-based string index, named the RLFM-index, 
that solves locate queries by executing a backward search algorithm on RLBWT. 
While the RLFM-index can solve locate queries in $O(r+n/s)$ words of space in $O((m + s \cdot occ) (\frac{\log \sigma}{\log \log r} + (\log \log n)^{2}))$ time 
for the number $r$ of runs in the RLBWT of $T$ and parameter $s \geq 1$, 
it consumes $O(r+n/s)$ words of space depending on the string length. 
Recently, Gagie et al.~\cite{10.1145/3375890} presented the r-index, which can reduce the space usage of the RLFM-index 
to one linearly proportional to the number of runs in RLBWT.
The r-index can solve locate queries space efficiently with only $O(r)$ words of space 
and in $O(m \log \log_{w} (\sigma + (n/r)) + occ \log \log_{w} (n/r) )$ time for a machine word size $w = \Theta(\log n)$.
If the r-index allows $O(r \log \log_{w} (\sigma + n/r))$ words of space to be used, 
it can solve locate queries in the optimal time, $O(m + occ)$. 
Although there are other important queries including count query, extract query, 
decompression and prefix search for various applications to string processing, 
no previous string index can support various queries in addition to locate queries based on RLBWT 
in an optimal time with only $O(r)$ words of space.
That is, developing a string index for various queries in an optimal time with $O(r)$ words of space remains a challenge.

In this paper, we present \emph{OptBWTR (optimal-time queries on BWT-runs compressed indexes)}, 
the first string index that supports various queries including locate, count, 
extract queries in optimal time and $O(r)$ words of space for the number $r$ of runs in RLBWT.  
LF and $\phi^{-1}$  are important functions for string indexes on RLBWT. 
The best previous data structure computes LF and $\phi^{-1}$ in $O(\log \log_{w} (n/r))$ time with $O(r)$ words of space. 
In this paper, we present a novel data structure that can compute LF and $\phi^{-1}$ in constant time and $O(r)$ words of space. 
Subsequently, we present OptBWTR that supports the following five queries in optimal time and $O(r)$ words of space. 
\begin{itemize}
\item {\bf Locate query:} 
OptBWTR can solve a locate query on an input string in $O(r)$ words of space and $O(m \log \log_{w} \sigma + occ)$ time, 
which is optimal for strings with polylogarithmic alphabets~(i.e., $\sigma = O(\polylog n)$).
\item {\bf Count query:} 
OptBWTR can return the number of occurrences of a query string on an input string in $O(r)$ words of space and $O(m \log \log_{w} \sigma)$ time, 
which is optimal for polylogarithmic alphabets. 
\item {\bf Extract query:} 
OptBWTR can return substrings starting at a given position bookmarked beforehand in a string in $O(1)$ time per character and $O(r + b)$ words of space, 
where $b$ is the number of bookmarked positions. 
Resolving extract queries is sometimes called the \emph{bookmarking problem}~\cite{DBLP:conf/latin/GagieGKNP14,DBLP:conf/spire/CordingGW16}.
\item {\bf Decompression:} 
OptBWTR decompresses the original string of length $n$ in optimal time~(i.e., $O(n)$). 
This is the first linear-time decompression algorithm for RLBWT in $O(r)$ words of working space. 
\item {\bf Prefix search:} 
OptBWTR can return the strings in a set $D$ that include a given pattern as their prefixes in 
optimal time~(i.e., $O(m + occ')$) and $O(r')$ words of space, where $occ'$ is the number of output strings 
and $r'$ is the number of runs in the RLBWT of a string made by concatenating the strings in $D$. 
\end{itemize}

The state-of-the-art string indexes for each type of query are summarized in Table~\ref{table:result}.

\renewcommand{\arraystretch}{0.7}
\begin{table}[t]
    \scriptsize
    \caption{
    Summary of space and time for (i) locate and (ii) count queries, (iii) extract queries (a.k.a the bookmarking problem), 
    (iv) decompression of BWT or RLBWT and (v) prefix searches for each query, 
    where $n$ is the length of the input string $T$, 
    $m$ is the length of a given string $P$, 
    $occ$ is the number of occurrences of $P$ in $T$, 
    $\sigma$ is the alphabet size of $T$, 
    $w = \Theta(\log n)$ is the machine word size, 
    $r$ is the number of runs in the RLBWT of $T$, 
    $s$ is a parameter, $g$ is the size of a compressed grammar deriving $T$, 
    $b$ is the number of input positions for the bookmarking problem, 
    $D$ is a set of strings of total length $n$, 
    $occ'$ is the number of strings in $D$ such that each string has $P$ as a prefix and 
    $r'$ is the number of runs in the RLBWT of a string made by concatenating the strings in $D$.}
    \label{table:result} 
    \center{	

    \begin{tabular}{r||c|c}
(i) Locate query & Space~(words) & Time \\ \hline
RLFM-index~\cite{DBLP:journals/jcb/MakinenNSV10} & $O(r+ n/s)$ & $O((m + s \cdot occ) (\frac{\log \sigma}{\log \log r} + (\log \log n)^{2}))$ \\ \hline
r-index~\cite{10.1145/3375890} & $O(r)$ & $O(m \log \log_{w} (\sigma + (n/r)) + occ \log \log_{w} (n/r))$ \\  
 & $O(r \log \log_{w} (\sigma + (n/r)))$ & $O(m + occ)$ \\  
 & $O(rw\log_{\sigma} \log_{w} n)$ & $O( \lceil m \log (\sigma) / w \rceil + occ )$ \\ \hline \hline
OptBWTR & $O(r)$ & $O(m \log \log_{w} \sigma + occ)$ \\ 
    \end{tabular}     
    \smallskip        

    \begin{tabular}{r||c|c}
(ii) Count query & Space~(words) & Time \\ \hline
RLFM-index~\cite{DBLP:journals/jcb/MakinenNSV10} & $O(r)$ & $O(m (\frac{\log \sigma}{\log \log r} + (\log \log n)^{2}))$ \\ \hline
r-index~\cite{10.1145/3375890} & $O(r)$ & $O(m \log \log_{w} (\sigma + (n/r)))$ \\  
 & $O(r \log \log_{w} (\sigma + (n/r)))$ & $O(m)$ \\ 
 & $O(rw\log_{\sigma} \log_{w} n)$ & $O( \lceil m \log (\sigma) / w \rceil )$ \\ \hline \hline
OptBWTR & $O(r)$ & $O(m \log \log_{w} \sigma)$ \\ 
    \end{tabular}     
    \smallskip        

    \begin{tabular}{r||c|c|c}
(iii) Extract query & Space~(words) & Time per character & Overhead \\ \hline
Gagie et al.\cite{DBLP:conf/latin/GagieGKNP14} & $O(g + b \log^{*} n )$ & $O(1)$ & - \\ \hline
Gagie et al.\cite{10.1145/3375890} & $O( r \log (n/r) )$ & $O(\log (\sigma) / w )$ & $O( \log (n/r))$ \\ \hline
Cording et al.\cite{DBLP:conf/spire/CordingGW16} & $O((g + b)\max\{ 1, \log^{*} g - \log^{*}( \frac{g}{b} - \frac{b}{g} )  \} )$ & $O(1)$ & - \\ \hline \hline
OptBWTR & $O(r + b)$ & $O(1)$ & - \\ 
    \end{tabular}     
    \smallskip        

    \begin{tabular}{r||c|c}
(iv) Decompression & Space~(words) & Time \\ \hline
Lauther and Lukovszki~\cite{DBLP:journals/algorithmica/LautherL10} & $O(n (\log \log n + \log \sigma)/w)$ & $O(n)$ \\ \hline
Golynski et al.\cite{DBLP:conf/soda/GolynskiMR06} & $O((n \log \sigma) / w)$ & $O(n \log \log \sigma)$ \\ \hline 
Predecessor queries~\cite{DBLP:journals/talg/BelazzouguiN15} & $O(r)$ & $O(n \log \log_{w} (n/r))$ \\ \hline \hline
OptBWTR & $O(r)$ & $O(n)$ \\ 
    \end{tabular}     
    \smallskip

    \begin{tabular}{r||c|c}
(v) Prefix search & Space~(words) & Time \\ \hline
Compact trie~\cite{DBLP:journals/jacm/Morrison68} & $(n \log \sigma)/w + O(|D|)$ & $O(m + occ')$ \\ \hline
Z-fast trie~\cite{DBLP:conf/spire/BelazzouguiBV10} & $(n \log \sigma)/w + O(|D|)$ & expected $O(\lceil m \log (\sigma) / w \rceil + \log m + \log \log \sigma + occ')$ \\ \hline
Packed c-trie~\cite{DBLP:journals/ieicet/TakagiISA17} & $(n \log \sigma)/w + O(|D|)$ & expected $O(\lceil m \log (\sigma) / w \rceil + \log \log n + occ')$ \\ \hline
c-trie++~\cite{DBLP:journals/corr/abs-1904-07467} & $(n \log \sigma)/w + O(|D|)$ & expected $O(\lceil m \log (\sigma) / w \rceil + \log \log_{\sigma} w + occ')$ \\ \hline \hline
OptBWTR & $O( r' + |D|)$ & $O(m + occ')$ \\ 
    \end{tabular}     
    \smallskip    
    }
\end{table}

This paper is organized as follows. In Section~\ref{sec:preliminary}, 
we introduce the important notions used in this paper. 
Section~\ref{sec:move_section} presents novel data structures for computing LF and $\phi^{-1}$ in constant time. 
Section~\ref{sec:backward_search} presents a data structure supporting a modified version of a backward search on RLBWT. 
The backward search leverages the two data structures introduced in Section~\ref{sec:move_section}.  
Sections~\ref{sec:OptBWTR} and \ref{sec:applications} present OptBWTR that supports all five queries mentioned above by leveraging the modified backward search, 
LF, and $\phi^{-1}$.  

%% file: 2_preliminaries.tex
\section{Preliminaries} \label{sec:preliminary}
Let $\Sigma = \{ 1, 2, \ldots, \sigma \}$ be an ordered alphabet of size $\sigma$, 
$T$ be a string of length $n$ over $\Sigma$ and $|T|$ be the length of $T$. 
Let $T[i]$ be the $i$-th character of $T$~(i.e., $T = T[1], T[2], \ldots, T[n]$) and $T[i..j]$ be the substring of $T$ that begins at position $i$ 
and ends at position $j$. For two strings, $T$ and $P$, $T \prec P$ means that $T$ is lexicographically smaller than $P$. 
Let $\varepsilon$ be the empty string, i.e., $|\varepsilon| = 0$. 
We assume that (i) $\sigma = n^{O(1)}$ and (ii)
the last character of string $T$ is a special character $\$$ not occurring on substring $T[1..n-1]$ such that 
$\$ \prec c$ holds for any character $c \in \Sigma \setminus \{ \$ \}$. 
For two integers, $b$ and $e$~($b \leq e$), \emph{interval} $[b, e]$ is the set $\{b, b+1, \ldots, e \}$. 
$\Occ(T, P)$ denotes all the occurrence positions of a string $P$ in a string $T$, i.e., $\Occ(T, P) = \{i \mid i \in [1, n-|P|+1] \mbox{ s.t. } P = T[i..(i+|P|-1)] \}$. 
A \emph{count query} on a string $T$ returns the number of occurrences of a given string $P$ in $T$, i.e., $|\Occ(T, P)|$. 
Similarly, a \emph{locate query} on string $T$ returns all the starting positions of $P$ in $T$, i.e., $\Occ(T, P)$. 

A \emph{rank} query $\rank(T, c, i)$ on a string $T$ returns the number of occurrences of a character $c$ in $T[1..i]$, i.e., $\rank(T, c, i) = |\Occ(T[1..i], c)|$. 
A \emph{select} query $\select(T, c, i)$ on a string $T$ returns the $i$-th occurrence of $c$ in $T$, i.e., 
the query returns the smallest integer $j \geq 1$ such that $|\Occ(T[1..j], c)| = i$. 
Assume that $T[b..e]$ contains a character $c$ for an interval $[b, e] \subseteq [1, n]$. 
Let $\hat{b}$ and $\hat{e}$ be the first and last occurrences of a character $c$ in $T[b..e]$~(i.e., 
$\hat{b} = \min \{ i \mid i \in [b, e] \mbox{ s.t. } T[i] = c \}$ and $\hat{e} = \max \{ i \mid i \in [b, e] \mbox{ s.t. } T[i] = c \}$).
Then, we can compute $\hat{b}$ and $\hat{e}$ by the following lemma. 
\begin{lemma}\label{lem:first_last_rank}
The following statements hold: 
(i) $T[b..e]$ contains a character $c$ if and only if $\rank(T, c, e) - \rank(T, c, b-1) \geq 1$ holds. 
(ii) $\hat{b} = \select(T, c, \rank(T, c, b-1) + 1)$ and $\hat{e} = \select(T, c, \rank(T, c, e))$ hold if $T[b..e]$ contains $c$. 
\end{lemma}

A suffix array ($\SA$) of a string $T$ is an integer array of size $n$ such that 
$\SA[i]$ stores the starting position of the $i$-th suffix of $T$ in lexicographical order. 
Formally, $\SA$ is a permutation of $[1, n]$ such that $T[\SA[1]..n] \prec \cdots \prec T[\SA[n]..n]$ holds. 
Each value in SA is called an \emph{sa-value}.

The \emph{suffix array interval}~(\emph{sa-interval}) of a string $P$ is an interval $[b,e] \subseteq [1,n]$ such that 
$\SA[b..e]$ represents all the occurrence positions of $P$ in string $T$; 
that is, for any integer $p \in [1, n]$, $T[p..p+|P|-1] = P$ holds if and only if $p \in \{ \SA[b], \SA[b+1], \ldots, \SA[e] \}$. 
The sa-interval of the empty string $\varepsilon$ is defined as $[1, n]$.

$\LF$ is a function that returns the position with sa-value $\SA[i]-1$ on $\SA$~(i.e., $\SA[\LF(i)] = \SA[i] - 1$) 
for a given integer $i \in [1, n]$ if $\SA[i] \neq 1$; otherwise, it returns the position with sa-value $n$~(i.e., $\SA[\LF(i)] = n$). 
$\phi^{-1}$~\cite{DBLP:conf/cpm/KarkkainenMP09} is a function that returns $\SA[i+1]$ for a given 
sa-value $\SA[i] \in [1, n]$~(i.e., $\phi^{-1}(\SA[i]) = \SA[i+1]$) if $i \neq n$; 
otherwise, it returns $\SA[1]$.

We will use base-2 logarithms throughout this paper unless indicated otherwise. 
Our computation model is a unit-cost word RAM with a machine word size of $w = \Theta(\log n)$ bits. 
We evaluate the space complexity in terms of the number of machine words. 
A bitwise evaluation of space complexity can be obtained with a $\log n$ multiplicative factor.

\subsection{BWT and run-length BWT~(RLBWT)}\label{sec:bwt}
\begin{wrapfigure}[18]{r}[1mm]{50mm}
\vspace{-2\baselineskip}
\begin{center}
		\includegraphics[width=0.3\textwidth]{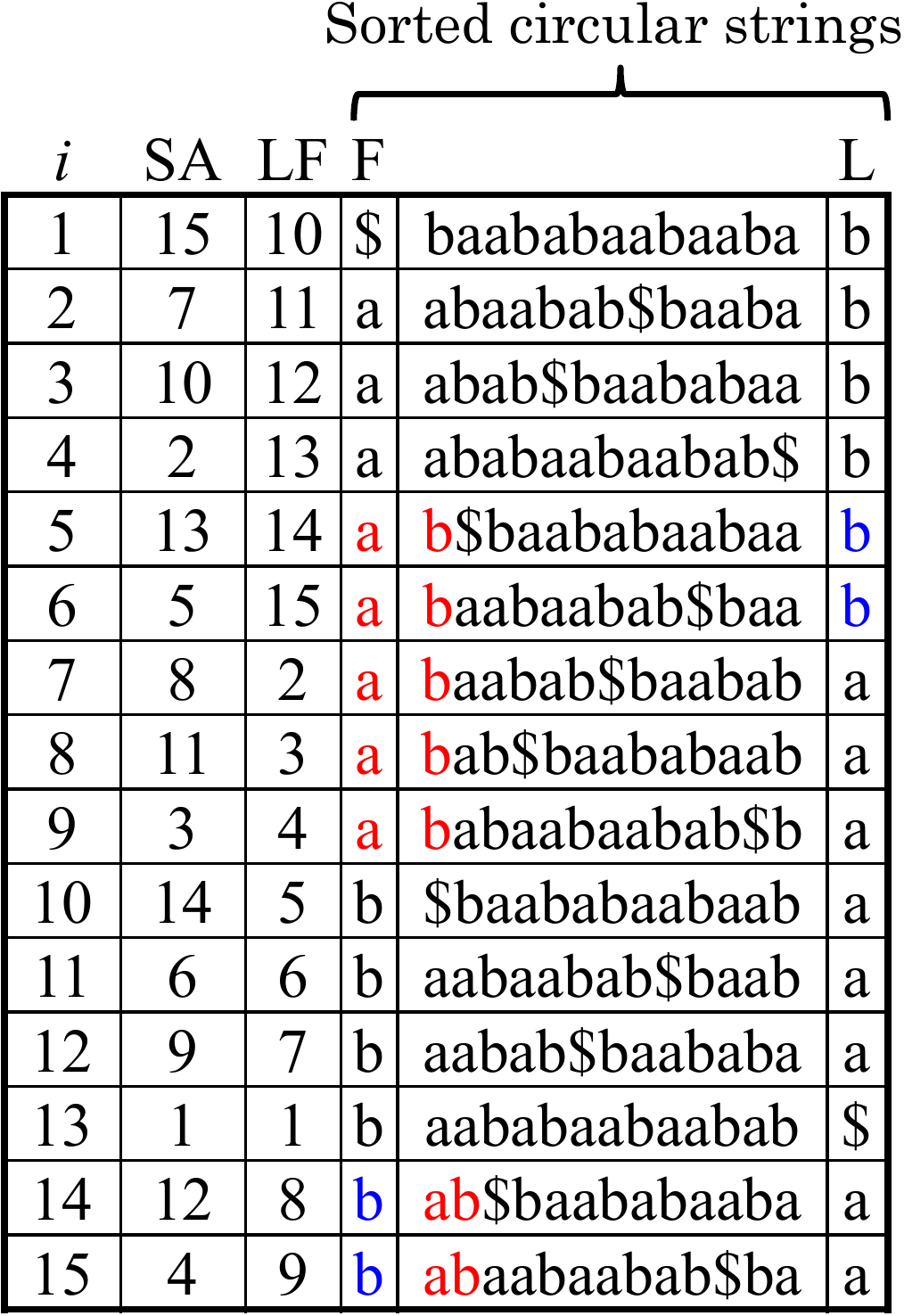}
\caption{Table illustrating the BWT~(L), SA, LF function, F, and  the sorted circular strings of $T = baababaabaabab\$$.}
 \label{fig:bwt}
\end{center}
\end{wrapfigure}

The BWT~\cite{burrows1994block} of a string $T$ is a string $L$ of length $n$ built by permuting $T$ as follows: 
(i) all $n$ circular strings of $T$~(i.e., $T[1..n]$, $T[2..n]T[1]$, $T[3..n]T[1..2]$, $\ldots$, $T[n]T[2..n-1]$) are sorted in lexicographical order; 
(ii) $L[i]$ is the last character at the $i$-th circular string in the sorted order for $i \in [1, n]$. 
Similarly, $F$ is a string of length $n$ such that $F[i]$ is the first character at the $i$-th circular string in the sorted order. 
Formally, let $L[i] = T[\SA[\LF(i)]]$ and $F[i] = T[\SA[i]]$. 

Let $C$ be an array of size $\sigma$ such that $C[c]$ is the number of occurrences of characters lexicographically smaller than $c\in \Sigma$ in string $T$ 
i.e., $C[c] = |\{ i \mid i \in [1, n] \mbox{ s.t. } T[i] \prec c \}|$. 
The BWT has the following property. 
For any integer $i \in [1, n]$, $\LF(i)$ is equal to the number of characters that are lexicographically smaller than the character $L[i]$ plus the rank of $L[i]$ on the BWT. 
Thus, $\LF(i) = C[c] + \rank(L, c, i)$ holds for $c = L[i]$. 
This is because $\LF(i) < \LF(j)$ if and only if either of the following conditions holds: (i) $L[i] \prec L[j]$ or (ii) $L[i] = L[j]$ and $i < j$ 
for two integers $1 \leq i < j \leq n$. 

Let $[b,e]$ be the sa-interval of a string $P$ and 
$[b', e']$ be the sa-interval of $cP$ for a character $c$. 
Then, the following relation holds between $[b,e]$ and $[b', e']$ on the BWT $L$. 
\begin{lemma}[e.g., \cite{DBLP:journals/jacm/FerraginaM05}]\label{lem:backward_search}
Let $\hat{b}$ and $\hat{e}$ be the first and last occurrences of $c$ in $L[b..e]$~(i.e., 
$\hat{b} = \min \{ i \mid i \in [b, e] \mbox{ s.t. } L[i] = c \}$ and $\hat{e} = \max \{ i \mid i \in [b, e] \mbox{ s.t. } L[i] = c \}$). 
Then, $b' = \LF(\hat{b})$, $e' = \LF(\hat{e})$, and $\SA[b'] = \SA[\hat{b}] - 1$ hold if $P$ and $cP$ are substrings of $T$.   
\end{lemma}

Figure~\ref{fig:bwt} illustrates the BWT, SA, LF function, $F$, $L$ and sorted circular strings of a string $T = baababaabaabab\$$. 
For example, let $P = ab$, $c = b$, and $T = baababaabaabab\$$. 
Then $[b,e] = [5, 9]$, $[b',e'] = [14, 15]$, $\hat{b} = 5$, and $\hat{e} = 6$~(see also Figure~\ref{fig:bwt}). 
Moreover, $b' = \LF(\hat{b})$ and $e' = \LF(\hat{e})$ hold by Lemma~\ref{lem:backward_search}.

The RLBWT of $T$ is a BWT encoded by run-length encoding; 
i.e., it is a partition of $L$ into $r$ substrings $\rlbwt(L) = L_{1}, L_{2}, \ldots, L_{r}$ 
such that each substring $L_{i}$ is a maximal repetition of the same character in $L$~(i.e., 
$L_{i}[1] = L_{i}[2] = \cdots = L_{i}[|L_{i}|]$ and 
$L_{i-1}[1] \neq L_{i}[1] \neq L_{i+1}[1]$).
Each $L_{i}$ is called a \emph{run}. 
Let $\ell_{i}$ be the starting position of the $i$-th run of BWT $L$, 
i.e., $\ell_{1} = 1$, $\ell_{i} = \ell_{i-1} + |L_{i-1}|$ for $i \in [2, r]$. 
Let $\ell_{r+1} = n+1$. 
The RLBWT is represented as $r$ pairs $(L_{1}[1], \ell_{1})$, $(L_{2}[1], \ell_{2})$, $\ldots$, $(L_{r}[1], \ell_{r})$ using $2r$ words.
For example, 
$\rlbwt(L) = bbbbbb, aaaaaa, \$, aa$ for BWT $L$ illustrated in Figure~\ref{fig:bwt}. 
The RLBWT is represented as $(b, 1), (a, 7), (\$, 13)$, and $(a, 14)$.

Let $\delta$ be a permutation of $[1,r]$ satisfying $\LF(\ell_{\delta[1]}) < \LF(\ell_{\delta[2]}) < \cdots < \LF(\ell_{\delta[r]})$. 
The LF function has the following properties on RLBWT. 
\begin{lemma}[e.g., Lemma 2.1 in \cite{DBLP:conf/soda/Kempa19}]\label{lem:LF_property}
The following two statements hold: 
(i) Let $x$ be the integer satisfying $\ell_{x} \leq i < \ell_{x+1}$ for some $i \in [1, n]$. 
Then, $\LF(i) = \LF(\ell_{x}) + (i - \ell_{x})$;  
(ii) $\LF(\ell_{\delta[1]}) = 1$ and 
$\LF(\ell_{\delta[i]}) = \LF(\ell_{\delta[i-1]}) + |L_{\delta[i-1]}|$ for all $i \in [2, r]$. 
\end{lemma}
\begin{proof}
See Appendix A.1.
\end{proof}

The sequence $u_{1}, u_{2}, \ldots$, $u_{r+1}$ consists of sa-values such that 
(i) $\{ u_{1}, u_{2}, \ldots$, $u_{r+1} \} = \{ \SA[\ell_{1} + |L_{1}| - 1], \SA[\ell_{2} + |L_{2}| - 1], \ldots, \SA[\ell_{r} + |L_{r}| - 1], n+1 \}$, 
and (ii) $u_{1} < u_{2} < \cdots < u_{r+1} = n+1$. 
Let $\delta'$ be a permutation of $[1,r]$ satisfying $\phi^{-1}(u_{\delta'[1]}) < \phi^{-1}(u_{\delta'[2]}) < \cdots < \phi^{-1}(u_{\delta'[r]})$. 
$\phi^{-1}$ has the following properties on RLBWT. 
\begin{lemma}[Lemma 3.5 in \cite{10.1145/3375890}]\label{lem:phi_property}
The following three statements hold: 
(i) 
Let $x$ be the integer satisfying $u_{x} \leq i < u_{x+1}$ for some integer $i \in [1, n]$. 
Then $\phi^{-1}(i) = \phi^{-1}(u_{x}) + (i - u_{x})$;  
(ii) 
$\phi^{-1}(u_{\delta'[1]}) = 1$ and $\phi^{-1}(u_{\delta'[i]}) = \phi^{-1}(u_{\delta'[i-1]}) + d$ for all $i \in [2, r]$, 
where $d = u_{\delta'[i-1]+1} - u_{\delta'[i-1]}$; 
(iii) $u_{1} = 1$.
\end{lemma}
\begin{proof}
See Appendix A.2.
\end{proof}
See also Appendix A.3 for examples of Lemmas~\ref{lem:LF_property} and \ref{lem:phi_property}. 

%% file: 3_block_permutation.tex
\section{Novel data structures for computing LF and \texorpdfstring{$\phi^{-1}$}{inverse phi} functions}\label{sec:move_section}
In this section, we present two new data structures  for computing LF and $\phi^{-1}$ functions in constant time with $O(r)$ words of space. 
Our key idea is to (i) divide the domains and ranges of two functions into at least $r$ non-overlapping intervals on RLBWT and 
(ii) compute two functions for each domain and range by a linear search in constant time.
First, we introduce a notion named \emph{disjoint interval sequence} that is used for a function with non-overlapping intervals for its domain and range. 
Then, we present a \emph{move query} for computing a function on each disjoint interval sequence and a novel data structure for efficiently computing move queries. 
Finally, we show that LF and $\phi^{-1}$ can be computed on two disjoint interval sequences using move queries. 

\subsection{Disjoint interval sequence and move query}\label{sec:move_function}
Let $I = (p_{1}, q_{1}), (p_{2}, q_{2}), \ldots, (p_{k}, q_{k})$ be a sequence of $k$ pairs of integers. 
We introduce a permutation $\pi$ of $[1,k]$ and sequence $d_{1}, d_{2}, \ldots, d_{k}$ for $I$. 
$\pi$ satisfies $q_{\pi[1]} \leq q_{\pi[2]} \leq \cdots \leq q_{\pi[k]}$, 
and $d_{i} = p_{i+1} - p_{i}$ for $i \in [1, k]$, 
where $p_{k+1} = n+1$. 
We call the sequence $I$ a \emph{disjoint interval sequence} if it satisfies the following three conditions: 
(i) $p_{1} = 1 < p_{2} < \cdots < p_{k} \leq n$ holds, 
(ii) $q_{\pi[1]} = 1$, 
and (iii) $q_{\pi[i]} = q_{\pi[i-1]} + d_{\pi[i-1]}$ holds for each $i \in [2, k]$. 

\begin{wrapfigure}[12]{r}[1mm]{60mm}
\vspace{-2\baselineskip}
\begin{center}
		\includegraphics[width=0.4\textwidth]{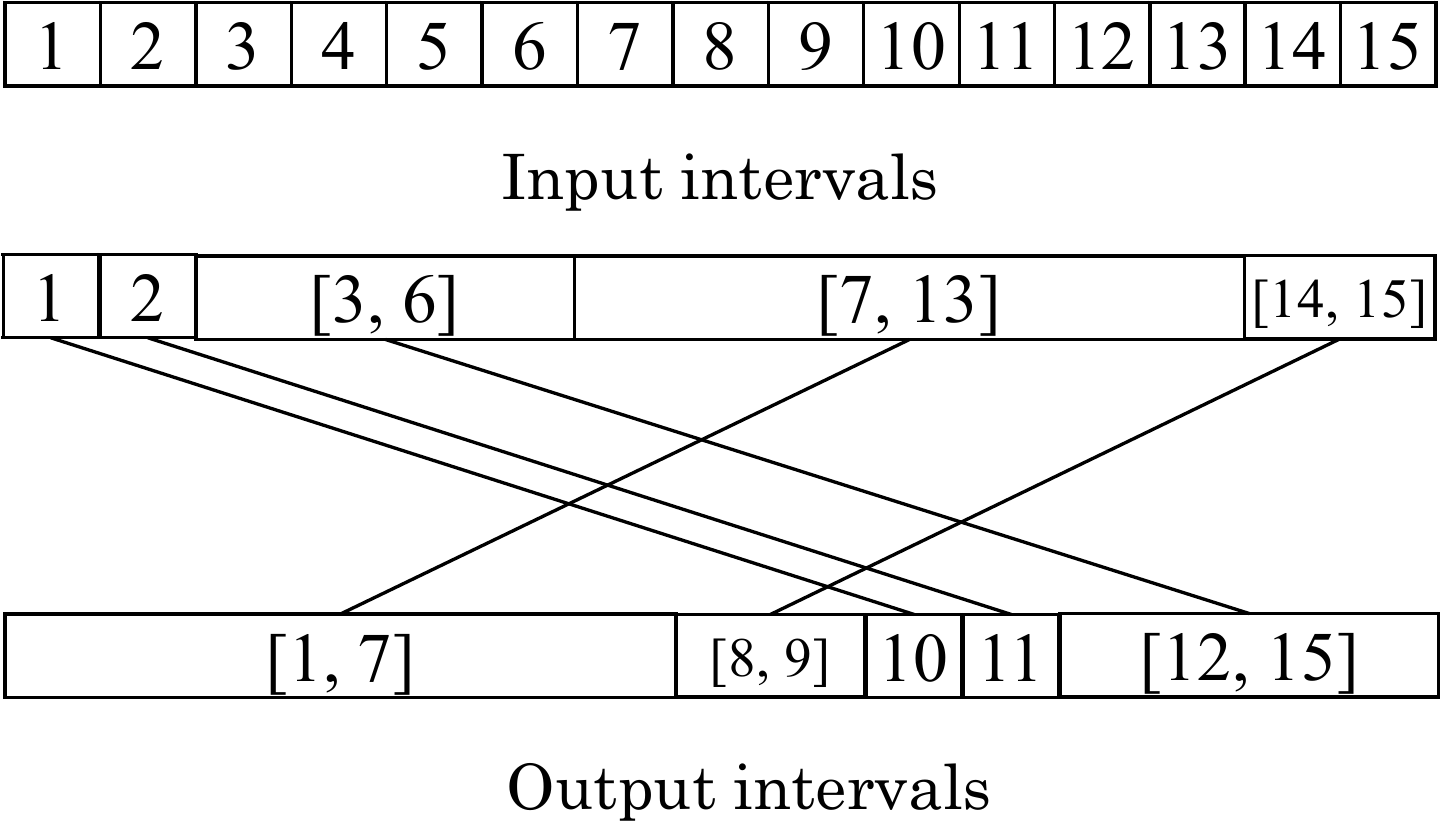}
\caption{Input and output intervals created by $I = (1, 10)$, $(2, 11)$, $(3, 12)$, $(7, 1)$, $(14, 8)$. 
The $i$-th input and output intervals are connected by a black line.}
 \label{fig:intervals}
\end{center}
\end{wrapfigure}

We call the two intervals $[p_{i}, p_{i} + d_{i} - 1]$ and $[q_{i}, q_{i} + d_{i} - 1]$ the $i$-th \emph{input and output intervals} 
of the disjoint interval sequence $I$, respectively, for each $i \in [1, k]$. 
The input intervals $[p_{1}, p_{1} + d_{1} - 1]$, $[p_{2}, p_{2} + d_{2} - 1]$, $\ldots$, $[p_{k}, p_{k} + d_{k} - 1]$ do not overlap, 
i.e., $[p_{i}, p_{i} + d_{i} - 1] \cap [p_{j}, p_{j} + d_{j} - 1] = \emptyset$ holds for any pair of two distinct integers $i, j \in [1, k]$.
Hence, the union of the input intervals is equal to the interval $[1, n]$, i.e., $\bigcup_{i=1}^{k} [p_{i}, p_{i} + d_{i} - 1] = [1, n]$. 
Similarly, the output intervals $[q_{1}, q_{1} + d_{1} - 1]$, $[q_{2}, q_{2} + d_{2} - 1]$, $\ldots$, $[q_{k}, q_{k} + d_{k} - 1]$) do not overlap,  
and their union is equal to $[1, n]$.

A move query $\movef(I, i, x)$ returns a pair $(i', x')$ on a disjoint interval sequence $I$ for a position $i \in [1, n]$ 
and the index $x$ of the input interval of $I$ containing the position $i$~(i.e., $x$ is the integer satisfying $i \in [p_{x}, p_{x}+d_{x}-1]$).
Here, $i' = q_{x} + (i - p_{x})$ and $x'$ is the index of the input interval of $I$ containing $i'$.
We can represent a bijective function using a disjoint interval sequence and move query. 
Formally, let $f_{I}(i) = i'$ for an integer $i \in [1, n]$, 
where $i'$ is the first value of the pair outputted by $\movef(I, i, x)$. 
$f_{I}$ maps the $j$-th input interval into the $j$-th output interval~(i.e., $f_{I}(i) = q_{j} + (i - p_{j})$ for $i \in [p_{j}, p_{j} + d_{j} - 1]$). 
Hence, $f_{I}$ is a bijective function from $[1, n]$ to $[1, n]$.

Figure~\ref{fig:intervals} illustrates the input and output intervals of 
the disjoint interval sequence $I = (1, 10)$, $(2, 11)$, $(3, 12)$, $(7, 1)$, $(14, 8)$, 
where $n = 15$. 
The input intervals created by $I$ are $[1,1], [2,2], [3, 6], [7,13]$, and $[14, 15]$. 
The output intervals created by $I$ are $[10, 10], [11,11], [12, 15], [1,7]$, and $[8, 9]$. 
For example, $\movef(I, 3, 3) = (12, 4)$, $\movef(I, 5, 3) = (14, 5)$, and $\movef(I, 8, 4) = (2, 2)$.

\subsection{Move data structure}\label{sec:graph_and_move_data_structure}
\begin{wrapfigure}[16]{r}[1mm]{60mm}
\vspace{-2\baselineskip}
\begin{center}
		\includegraphics[width=0.35\textwidth]{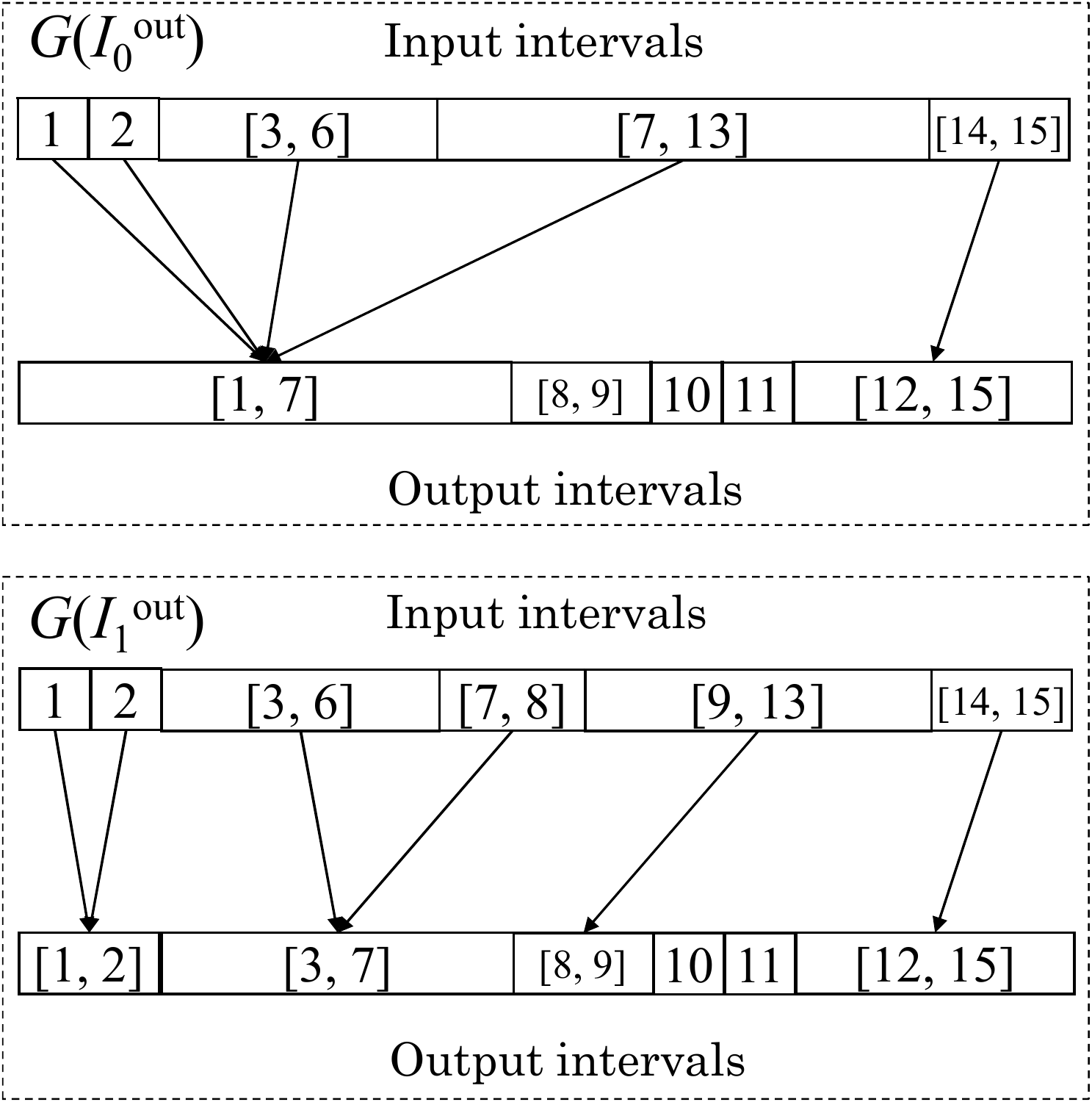}
\caption{
Two permutation graphs $G(I^{\mathsf{out}}_{0})$ and $G(I^{\mathsf{out}}_{1})$ for $I$. 
Here, $I$ is the disjoint interval sequence illustrated in Figure~\ref{fig:intervals}. 
}
 \label{fig:permutation_graph}
\end{center}
\end{wrapfigure}
In this section, we present a data structure called \emph{move data structure} for computing move queries in constant time. 
To do so, we introduce three notions, i.e., the \emph{permutation graph}, \emph{split interval sequence}, and \emph{balanced interval sequence}. 
A permutation graph $G(I)$ is a directed graph for a disjoint interval sequence $I$. 
The number of nodes in $G(I)$ is $2k$, 
and the nodes correspond one-by-one with the input and output intervals of $I$. 
Each input interval $[p_{i}, p_{i} + d_{i} - 1]$ has a single outgoing edge pointing to 
the output interval $[q_{j}, q_{j} + d_{j} - 1]$ containing $p_{i}$; i.e., $j$ is the integer satisfying $p_{i} \in [q_{j}, q_{j} + d_{j} - 1]$. 
Hence, $G(I)$ has $k$ edges. 
We say that $I$ is \emph{out-balanced} if 
every output interval has at most three incoming edges. 

A split interval sequence $I^{\mathsf{out}}_{t}$ is a disjoint interval sequence for a disjoint interval sequence $I$ and an integer $t \geq 0$. 
Let $I^{\mathsf{out}}_{0} = I$. 
For $t \geq 1$, 
we define $I^{\mathsf{out}}_{t}$ using $I^{\mathsf{out}}_{t-1}$ and two integers $j, d$  if $I^{\mathsf{out}}_{t-1}$ is not out-balanced. 
Let (i) $I^{\mathsf{out}}_{t-1} = (p'_{1}, q'_{1}), (p'_{2}, q'_{2}), \ldots, (p'_{k'}, q'_{k'})$, 
(ii) $j$ be the smallest integer such that 
the $j$-th output interval of $I^{\mathsf{out}}_{t-1}$ has at least four incoming edges in $G(I^{\mathsf{out}}_{t-1})$, 
and (iii) $d$ be the largest integer satisfying $|[q_{j}, q_{j} + d -1] \cap \{ p_{1}, p_{2}, \ldots, p_{k'} \}| = 2$.
Then, $I^{\mathsf{out}}_{t}$ is defined as $(p'_{1}, q'_{1})$, $(p'_{2}, q'_{2})$, $\ldots$, $(p'_{j-1}, q'_{j-1})$, 
$(p'_{j}, q'_{j})$, $(p'_{j}+d, q'_{j}+d)$, $\ldots$, $(p'_{k'}, q'_{k'})$. 
In other words, 
$I^{\mathsf{out}}_{t}$ is created by splitting the $j$-th pair $(p'_{j}, q'_{j})$ of $I^{\mathsf{out}}_{t-1}$ 
into two pairs $(p'_{j}, q'_{j})$ and $(p'_{j}+ d, q'_{j} + d)$.
Let $\tau \geq 0$ be the smallest integer such that $I^{\mathsf{out}}_{\tau}$ is out-balanced. 

Figure~\ref{fig:permutation_graph} illustrates two permutation graphs $G(I^{\mathsf{out}}_{0})$ and $G(I^{\mathsf{out}}_{1})$, 
where $I$ is the disjoint interval sequence illustrated in Figure~\ref{fig:intervals}, i.e., 
$I = (1, 10)$, $(2, 11)$, $(3, 12)$, $(7, 1)$, $(14, 8)$. 
The fourth output interval $[1, 7]$ of $I^{\mathsf{out}}_{0}$ has four incoming edges, 
and the other output intervals have at most one incoming edge in $G(I^{\mathsf{out}}_{0})$. 
Hence, $I^{\mathsf{out}}_{1} = (1, 10)$, $(2, 11)$, $(3, 12)$, $(7, 1)$, $(9, 3)$, $(14, 8)$ holds by $j=4$ and $d=2$. 
$I^{\mathsf{out}}_{1}$ is out-balanced, and hence $\tau = 1$ holds.

The split interval sequence has the following four properties for each $t \in [0, \tau]$: 
(i) $I^{\mathsf{out}}_{t}$ consists of $k + t$ pairs. 
(ii) $I^{\mathsf{out}}_{t}$ consists of at least $2t$ pairs. 
(iii) 
Let $d'_{i} = p'_{i+1} - p'_{i}$ for $i \in [1, k']$ and $p'_{k'+1} = n+1$. 
Both output intervals $[q'_{j}, q'_{j} + d -1]$ and $[q'_{j} + d, q'_{j} + d'_{j} -1]$ have at least two incoming edges in $G(I^{\mathsf{out}}_{t})$.  
(iv) Let $f_{I}$ and $f^{t}_{I}$ be the two bijective functions represented by $I$ and $I^{\mathsf{out}}_{t}$, respectively. 
Then, $f_{I}(i) = f^{t}_{I}(i)$ holds for $i \in [1,n]$. 
Formally, we obtain the second property from the following lemma. 
\begin{lemma}\label{lem:balance_size}
$|I^{\mathsf{out}}_{t}| \geq 2t$ holds for any $t \in [0, \tau]$.
\end{lemma}
\begin{proof}
See Appendix B.1.
\end{proof}

A balanced interval sequence $B(I)$ is defined as $I^{\mathsf{out}}_{\tau}$ for a disjoint interval sequence $I$. 
We obtain the lemma below from the four properties of $I^{\mathsf{out}}_{\tau}$. 
\begin{lemma}\label{lem:balanced_sequence2}
Let $f_{I}$ and $f_{B(I)}$ be the two bijective functions represented by $I$ and $B(I)$, respectively for a disjoint interval sequence $I$ of length $k$. 
The following three statements hold: 
(i) $|B(I)| \leq 2k$, (ii) $B(I)$ is out-balanced, 
and (iii) the two disjoint interval sequences $I$ and $B(I)$ represent the same bijective function, i.e., 
$f_{I}(i) = f_{B(I)}(i)$ for $i \in [1,n]$.   
\end{lemma}
\begin{proof}
(i) We obtain an inequality $\tau \leq k$ from the first and second properties of $I^{\mathsf{out}}_{t}$, 
because $k + t \geq 2t$ must hold for any $t \in [0, \tau]$. 
Hence, $I^{\mathsf{out}}_{\tau}$ consists of at most $2k$ pairs; i.e., $|B(I)| \leq 2k$ holds.
\end{proof}

The move data structure $F(I)$ is built on a balanced interval sequence 
$B(I) = (p_{1}, q_{1})$, $(p_{2}, q_{2})$, $\ldots$, $(p_{k'}, q_{k'})$ for a disjoint interval sequence $I$, 
and it supports move queries on $B(I)$. 
The move data structure consists of two arrays $D_{\mathsf{pair}}$ and $D_{\mathsf{index}}$ of size $k'$. 
$D_{\mathsf{pair}}[i]$ stores the $i$-th pair $(p_{i}, q_{i})$ of $B(I)$ for each $i \in [1, k']$. 
$D_{\mathsf{index}}[i]$ stores the index $j$ of the input interval containing $q_{i}$.  
Hence, the space usage is $3k'$ words in total. 

Now let us describe an algorithm for solving a move query $\movef(B(I), i, x) = (i', x')$ on $B(I)$, 
where $x$ and $x'$ are the indexes of the two input intervals of $B(I)$ containing $i$ and $i'$, respectively, 
and $i' = q_{x} + (i - p_{x})$. 
The algorithm consists of three steps. 
In the first step, the algorithm computes $i' = q_{x} + (i - p_{x})$. 
In the second step, 
the algorithm finds the $x'$-th input interval by a linear search on the input intervals of $B(I)$. 
Let $b = D_{\mathsf{index}}[x]$.
The linear search starts at the $b$-th input interval $[p_{b}, p_{b+1}-1]$, 
reads the input intervals in the left-to-right order, 
and stops if the input interval containing position $i'$ is found~(i.e., the $x'$-th input interval). 
The linear search is always successful~(i.e., $x' \geq b$), 
because $i' \geq q_{x}$ holds. 
In the third step, the algorithm returns the pair $(i', x')$. 
The running time of the algorithm is $O(x' - b + 1)$ in total. 

The running time is computed as follows. 
Let $i_{\mathsf{beg}}$ and $i_{\mathsf{end}}$ be the indexes of the first and last input intervals that are connected to the $x$-th output interval in $G(B(I))$. 
The $x$-th output interval has at most three incoming edges, 
and hence, $i_{\mathsf{end}} - i_{\mathsf{beg}} + 1 \leq 3$ holds. 
Since $b$ is the index of an input interval that overlaps the $x$-th output interval, 
$i_{\mathsf{beg}} - 1 \leq b \leq i_{\mathsf{end}}$. 
Similarly, $i_{\mathsf{beg}} - 1 \leq x' \leq i_{\mathsf{end}}$. 
Therefore, $x' - b \leq 3$ and 
we can solve the move query in constant time. 

\subsection{Computing LF and \texorpdfstring{$\phi^{-1}$}{inverse phi} functions using move data structures}\label{sec:lf_phi_move}
Here, we show that we can compute the LF function using a move data structure. 
Recall that 
$\ell_{i}$ is the starting position of the $i$-th run on BWT $L$ for $i \in [1, r]$, 
and $\delta$ is the permutation of $[1, r]$ introduced in Section~\ref{sec:bwt}.
The sequence $I_{\mathsf{LF}}$ is defined as 
$r$ pairs $(\ell_{1}, \LF(\ell_{1}))$, $(\ell_{2}, \LF(\ell_{2}))$, $\ldots$, $(\ell_{r}, \LF(\ell_{r}))$. 
$I_{\mathsf{LF}}$ satisfies the three conditions of a disjoint interval sequence by Lemma~\ref{lem:LF_property}, i.e., 
(i) $\ell_{1} = 1 < \ell_{2} < \cdots < \ell_{r} \leq n$, 
(ii) $\LF(\ell_{\delta[1]}) = 1$, 
and (iii) $\LF(\ell_{\delta[i]}) = \LF(\ell_{\delta[i-1]}) + |L_{\delta[i-1]}|$ holds for each $i \in [2, r]$. 
Hence $I_{\mathsf{LF}}$ is a disjoint interval sequence.

Let $f_{\mathsf{LF}}$ be the bijective function represented by the disjoint interval sequence $I_{\mathsf{LF}}$. 
Then, $f_{\mathsf{LF}}(i) = \LF(\ell_{x}) + (i - \ell_{x})$ holds, 
where $x$ is the integer such that $\ell_{x} \leq i < \ell_{x+1}$ holds. 
On the other hand, we have $\LF(i) = \LF(\ell_{x}) + (i - \ell_{x})$ by Lemma~\ref{lem:LF_property}(i). 
Hence, $f_{\mathsf{LF}}$ and LF are the same function, i.e., $\LF(i) = f_{\mathsf{LF}}(i)$ for $i \in [1, n]$.  

Let $F(I_{\mathsf{LF}})$ be the move data structure built on the balanced interval sequence $B(I_{\mathsf{LF}})$ for $I_{\mathsf{LF}}$. 
By Lemma~\ref{lem:balanced_sequence2}, 
the move data structure requires $O(r)$ words of space, 
and $\LF(i) = i'$ holds for a move query $\movef(B(I), i, x) = (i', x')$ on $B(I_{\mathsf{LF}})$. 
Hence, we have proven the following theorem. 
\begin{theorem}\label{theo:lf_inv_theorem}
Let $x$ and $x'$ be the indexes of the two input intervals of $B(I_{\mathsf{LF}})$ containing an integer $i \in [1, n]$ and $\LF(i)$, respectively. 
We can compute $\LF(i)$ and $x'$ in constant time 
by using $F(I_{\mathsf{LF}})$ and $(i, x)$. 
\end{theorem}

Similarly, we can show that we can compute $\phi^{-1}$ by using a move data structure. 
A sequence $I_{\mathsf{SA}}$ consists of $r$ pairs $(u_{1}, \phi^{-1}(u_{1}))$, $(u_{2}, \phi^{-1}(u_{2}))$, $\ldots$, $(u_{r}$, $\phi^{-1}(u_{r}))$, 
where $u_{1}, u_{2}, \ldots, u_{r}$ are the integers introduced in Section~\ref{sec:bwt}. 
$I_{\mathsf{SA}}$ satisfies the three conditions of a disjoint interval sequence by Lemma~\ref{lem:phi_property},  
and $\phi^{-1}$ is equal to the bijective function represented by $I_{\mathsf{SA}}$~(see Appendix B.2 for details). 
Let $F(I_{\mathsf{SA}})$ be the move data structure built on the balanced interval sequence $B(I_{\mathsf{SA}})$ for $I_{\mathsf{SA}}$. 
Then, the result of a move query on $B(I_{\mathsf{SA}})$ contains $\phi^{-1}(i)$ for $i \in [1, n]$, 
and hence, we have proven (i) of the following theorem. 
\begin{theorem}\label{theo:phi_inv_theorem}
Let $x$, $x'$, $\hat{x}$ be the indexes of the three input intervals of $B(I_{\mathsf{SA}})$ containing an integer $i \in [1, n]$, $\phi^{-1}(i)$, and $i-1$, respectively. 
Then, the following two statements hold: 
(i) We can compute $\phi^{-1}(i)$ and $x'$ in constant time using data structure $F(I_{\mathsf{SA}})$ and the pair $(i, x)$. 
(ii) We can compute the index $\hat{x}$ using $F(I_{\mathsf{SA}})$ and $(i, x)$.
\end{theorem}
\begin{proof}
(ii) 
Let $B(I_{\mathsf{SA}}) = (p_{1}, q_{1})$, $(p_{2}, q_{2})$, $\ldots$, $(p_{k'}, q_{k'})$. 
The $x$-th input interval is $[p_{x}, p_{x+1}-1]$, which contains $i$. 
$\hat{x} = x$ holds if $p_{x} \neq i$; otherwise, $\hat{x} = x-1$. 
We can verify $p_{x} \neq i$ holds in constant time by using $F(I_{\mathsf{SA}})$. 
\end{proof}

Theorems~\ref{theo:lf_inv_theorem} and \ref{theo:phi_inv_theorem} indicate that 
we can compute the position obtained by recursively applying LF and $\phi^{-1}$ to a position $i \in [1, n]$ $t$ times in $O(t)$ time 
if we know the index of the input interval containing $i$. 
For example, let $x, x', x''$, and $x'''$ be the indexes of the four input intervals of $B(I_{\mathsf{LF}})$ 
containing $i$, $\LF(i)$, $\LF(\LF(i))$, and $\LF(\LF(\LF(i)))$, respectively. 
$\LF(\LF(\LF(i)))$ can be computed by computing three move queries $\movef(B(I_{\mathsf{LF}}), i, x) = (\LF(i), x')$, $\movef(B(I_{\mathsf{LF}}), \LF(i), x') = (\LF(\LF(i)), x'')$, and $\movef(B(I_{\mathsf{LF}}), \LF(\LF(i)), x'') = (\LF(\LF(\LF(i))), x''')$.

%% file: 4_backward_search.tex
\section{New data structure for backward searches}\label{sec:backward_search}
Here, we present a modified version of the \emph{backward search}~\cite{DBLP:journals/jacm/FerraginaM05,DBLP:journals/tcs/BannaiGI20}, 
which we call \emph{backward search query for OptBWTR}~(BSR query), 
for computing the sa-interval of $cP$ for a given string $P$ and character $c$. 
To define the BSR query, we will introduce a new tuple: a \emph{balanced sa-interval} of a string $P$ is a 6-tuple $(b, e, \SA[b], i, j, v)$. 
Here, (i) $[b,e]$ is the sa-interval of $P$;  
(ii) $i$ and $j$ are the indexes of the two input intervals of $B(I_{\mathsf{LF}})$ containing $b$ and $e$, respectively;  
(iii) $v$ is the index of the input interval of $B(I_{\mathsf{SA}})$ containing $\SA[b]$. 
The balanced sa-interval of $P$ is undefined if the sa-interval of $P$ is $\emptyset$~(i.e., $P$ is not a substring of $T$). 
The input of the BSR query is the balanced sa-interval $(b, e, \SA[b], i, j, v)$ of a string $P$ and a character $c$. 
The output of the BSR query is the balanced sa-interval $(b', e', \SA[b'], i', j', v')$ of string $cP$ if the sa-interval of $cP$ is not the empty set; 
otherwise BSR outputs a mark $\perp$. 

\begin{wrapfigure}[12]{r}[1mm]{60mm}
\vspace{-2\baselineskip}
\begin{center}
		\includegraphics[width=0.4\textwidth]{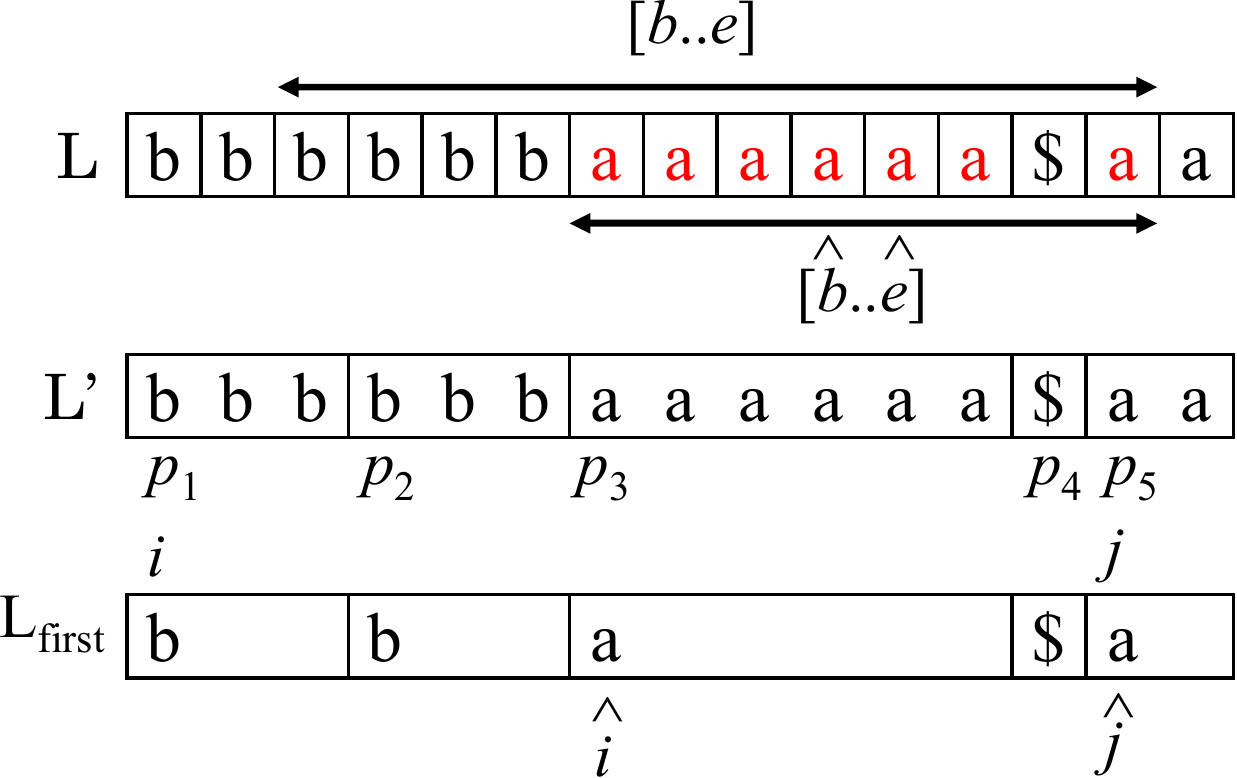}
\caption{Example of modified toehold lemma. }
 \label{fig:bsr}
\end{center}
\end{wrapfigure}

Now, we will present a data structure called the \emph{BSR data structure}. 
The BSR data structure supports BSR queries in $O(\log \log_{w} \sigma)$ time. 
It consists of five data structures $F(I_{\mathsf{LF}})$,
$F(I_{\mathsf{SA}})$, $R(\sbwt)$, $\SA^{+}$, and $\SA^{+}_{\mathsf{index}}$. 
Here, $F(I_{\mathsf{LF}})$ and $F(I_{\mathsf{SA}})$ are 
the two move data structures introduced in Section~\ref{sec:lf_phi_move}. 
Let $B(I_{\mathsf{LF}}) = (p_{1}, q_{1}), (p_{2}, q_{2}), \ldots, (p_{k}, q_{k})$. 
Then $\sbwt$ is the string satisfying $\sbwt = L[p_{1}], L[p_{2}], \ldots, L[p_{k}]$. 
$R(\sbwt)$ is a data structure called the \emph{rank-select data structure} built on $\sbwt$. 
$R(\sbwt)$ requires $O(|\sbwt|)$ words of space, 
and it supports rank and select queries on $\sbwt$ in $O(\log\log_{w} \sigma)$ and $O(1)$ time, respectively. 
See Appendix C.1 for the details of $R$. 
$\SA^{+}$ is an array of size $k$ such that 
$\SA^{+}[x]$ stores the sa-value at the starting position of the $x$-th input interval of $B(I_{\mathsf{LF}})$ for each $x \in [1, k]$~(i.e., 
$\SA^{+}[x] = \SA[p_{x}]$). 
Let $B(I_{\mathsf{SA}}) = (p'_{1}, q'_{1}), (p'_{2}, q'_{2}), \ldots, (p'_{k'}, q'_{k'})$. 
$\SA^{+}_{\mathsf{index}}$ is an array of size $k$ such that 
$\SA^{+}_{\mathsf{index}}[x]$ stores the index $y$ of the input interval of $B(I_{\mathsf{SA}})$ containing the position $\SA^{+}[x]$~(i.e., 
$y$ is the integer satisfying $\SA^{+}[x] \in [p'_{y}, p'_{y+1} - 1]$).
The space usage of the five data structures is $O(|B(I_{\mathsf{LF}})| + |B(I_{\mathsf{SA}})|)$ words, 
and $|B(I_{\mathsf{LF}})|, |B(I_{\mathsf{SA}})| = O(r)$ holds by Lemma~\ref{lem:balanced_sequence2}(i).

Next, we will present a key observation on BSR queries, which is based on the toehold lemma~(see, e.g., 
\cite{DBLP:journals/algorithmica/PolicritiP18,10.1145/3375890,DBLP:journals/tcs/BannaiGI20}).
Let $L'$ be a sequence of $k$ substrings $L[p_{1}..p_{2}-1], L[p_{2}..p_{3}-1], \ldots, L[p_{k}..p_{k+1}-1]$ of BWT $L$, 
where $p_{k+1} = n+1$. 
Then, $L'$ has the following properties: 
(i) $L'$ represents a partition of $L$.  
(ii) Each string of $L'$ consists of a repetition of the same character.  
(iii) Each character $\sbwt[t]$ corresponds to the first character of the $t$-th string of $L'$. 
(iv) The $i$-th and $j$-th strings of $L'$ contain the $b$-th and $e$-th characters of BWT $L$, respectively. 
(v) 
Let $\hat{b}$ and $\hat{e}$ be the first and last occurrences of $c$ in $L[b..e]$~(i.e., 
$\hat{b} = \min \{ t \mid t \in [b, e] \mbox{ s.t. } L[t] = c \}$ and $\hat{e} = \max \{ t \mid t \in [b, e] \mbox{ s.t. } L[t] = c \}$). 
Similarly, let $\hat{i}$ and $\hat{j}$ be the indexes of the two strings of $L'$ containing 
the $\hat{b}$-th and $\hat{e}$-th characters of BWT $L$, respectively. 
Then $\hat{i}$ and $\hat{j}$ are equal to the first and last occurrences of $c$ in $\sbwt[i..j]$. 
We obtain the following relations among the four positions $b$, $\hat{b}$, $e$, and $\hat{e}$ by using the above five properties: 
(i) $\hat{b} = b$ if $\sbwt[i] = c$; 
otherwise, $\hat{b} = p_{\hat{i}}$. 
(ii) Similarly, $\hat{e} = e$ if $\sbwt[j] = c$; otherwise $\hat{b} = p_{\hat{j}+1} - 1$. 
We call these two relations the \emph{modified toehold lemma}.

Let $\hat{v}$ be the index of the input interval of $B(I_{\mathsf{SA}})$ containing position $\SA[\hat{b}]$. 
By the modified toehold lemma, $\hat{v} = v$ and $\SA[\hat{b}] = \SA[b]$ hold if $\hat{b} = b$; 
otherwise, $\hat{v} = \SA^{+}_{\mathsf{index}}[\hat{i}]$ and $\SA[\hat{b}] = \SA^{+}[\hat{i}]$. 
We can compute the balanced sa-interval of $cP$ by using $F(I_{\mathsf{LF}})$ and $F(I_{\mathsf{SA}})$ after computing the six integers 
$\hat{b}, \hat{e}, \hat{i}, \hat{j}, \hat{v}, \SA[\hat{b}]$, 
because $b' = \LF(\hat{b})$, $e' = \LF(\hat{e})$, and $\SA[b'] = \SA[\hat{b}] - 1$ hold by Lemma~\ref{lem:backward_search}. 

Figure~\ref{fig:bsr} illustrates an example of the modified toehold lemma for a BWT $L = bbbbbbaaaaaa\$aa$. 
In this example, (i) $k = 5$, (ii) $(p_{1}, p_{2}, p_{3}, p_{4}, p_{5}) = (1, 4, 7, 13, 14)$, (iii) $c = a$,  
(iv) $L' = bbb, bbb, aaaaaa, \$, aa$, (v) $\sbwt = bba\$a$, and (vi) $(b, e, \hat{b}, \hat{e}, i, j, \hat{i}, \hat{j}) = (3, 14, 7, 14, 1, 5, 3, 5)$. 
The $i$-th string of $L'$ is not a repetition of the character $c$, and 
the $\hat{i}$-th string of $L'$ contains the $\hat{b}$-th character of $L$. 
Hence $\hat{b} = p_{\hat{i}} = 7$ holds by the modified toehold lemma. 
Similarly, the $j$-th string of $L'$ is a repetition of $c$, 
and hence $\hat{e} = e$ holds by the modified toehold lemma. 

We solve a BSR query in four steps. 
In the first step, we verify whether 
$\sbwt[i..j]$ contains character $c$ by computing two rank queries $\rank(\sbwt, c, j)$ and $\rank(\sbwt, c, i)$. 
By Lemma~\ref{lem:first_last_rank}(i), $\sbwt[i..j]$ contains $c$ if $\rank(\sbwt, c, j) - \rank(\sbwt, c, i) \geq 1$; 
otherwise, $cP$ is not a substring of $T$, and hence BSR outputs a mark $\perp$. 
In the second step, we compute two integers $\hat{i}$ and $\hat{j}$ using rank and select queries on the string $\sbwt$. 
$\hat{i} = \select(\sbwt, c, \rank(\sbwt, c, i-1) + 1)$ and $\hat{j} = \select(\sbwt, c, \rank(\sbwt, c, j))$ hold by Lemma~\ref{lem:first_last_rank}(ii). 
In the third step, 
we compute $\hat{b}$, $\hat{e}$, $\hat{v}$, and $\SA[\hat{b}]$ by the modified toehold lemma. 
In the fourth step, we compute the balanced sa-interval of $cP$ by processing the six integers $\hat{b}, \hat{e}, \hat{i}, \hat{j}, \hat{v}, \SA[\hat{b}]$, 
i.e., 
we compute (i) the pair $(b', i')$ using a move query on $B(I_{\mathsf{LF}})$ for the pair $(\hat{b}, \hat{i})$, 
(ii) the pair $(e', j')$ using a move query on $B(I_{\mathsf{LF}})$ for the pair $(\hat{e}, \hat{j})$, 
and (iii) the pair $(\SA[v'], v')$ by Theorem~\ref{theo:phi_inv_theorem}(ii). 
The running time is $O(\log \log_{w} \sigma)$ in total.

%% file: 4_r-index-f.tex
\section{OptBWTR}\label{sec:OptBWTR}
Here, we present OptBWTR, which supports optimal-time queries for polylogarithmic alphabets by leveraging data structures for computing LF and $\phi^{-1}$ functions. 
Let $P$ be a string of length $m$ in a count or locate query and $occ = |\Occ(T, P)|$. 
The goal of this section is to prove the following theorem. 
\begin{theorem}\label{theo:OptBWTR}
OptBWTR requires $O(r)$ words, 
and it supports count and locate queries on a string $T$ in $O(m \log \log_{w} \sigma)$ and $O(m \log \log_{w} \sigma + occ)$ time, respectively. 
We can construct OptBWTR in $O(n + r \log r)$ time and $O(r)$ words by processing the RLBWT of $T$.
\end{theorem}
\begin{proof}
See Appendix D.1 for the proof of the construction time and working space in Theorem~\ref{theo:OptBWTR}.
\end{proof}

OptBWTR consists of the five data structures composing the BSR data structure, 
i.e., $F(I_{\mathsf{LF}})$, $F(I_{\mathsf{SA}})$, $R(\sbwt)$, $\SA^{+}$, and $\SA^{+}_{\mathsf{index}}$. 
First, we present an algorithm for a count query using OptBWTR that consists of two phases. 
In the first phase, 
the algorithm computes the balanced sa-interval of $P$ by iterating BSR query $m$ times. 
The input of the $i$-th BSR query is the $(m-i+1)$-th character of $P$~(i.e., $P[m-i+1]$) and the balanced sa-interval of $P[m-i+2..m]$ for each $i \in [1, m]$. 
Here, $P[m+1..m]$ is defined as the empty string $\varepsilon$. 
The balanced sa-interval of $\varepsilon$ is $(1, n, n, 1, |B(I_{\mathsf{LF}})|, |B(I_{\mathsf{SA}})|)$, 
because (i) the sa-interval of the empty string is $[1, n]$, 
and (ii) $\SA[1] = n$ holds because $T$ contains the special character $\$$. 
The $i$-th BSR query outputs the sa-interval of $P[m-i+1..m]$ if $P[m-i+1..m]$ is a substring of $T$; 
otherwise it outputs a mark $\perp$. 
If a BSR query outputs $\perp$, 
the pattern $P$ does not occur in $T$. 
In this case, the algorithm stops and returns $0$ as the solution for the count query.
In the second phase, the algorithm returns the length of the sa-interval $[b,e]$ of $P$~(i.e., $e -b + 1$) as the solution for the count query, 
because $occ = e-b+1$ holds. 
The sa-interval of $P$ is contained in the balanced sa-interval of $P$; 
hence, the running time is $O(m \log \log_{w} \sigma)$ in total. 

Next, we present an algorithm for a locate query using OptBWTR. 
Let $v_{t}$ be the index of the input interval of $B(I_{\mathsf{SA}})$ containing $\SA[b+t]$ for $t \in [0, e-b]$. 
Then $\SA[b+1..e]$ can be computed by computing $(e-b)$ move queries $\movef(B(I_{\mathsf{SA}}), \SA[b], v_{0}) = (\SA[b+1], v_{1})$, 
$\movef(B(I_{\mathsf{SA}}), \SA[b+1], v_{1}) = (\SA[b+2], v_{2})$, $\ldots$, $\movef(B(I_{\mathsf{SA}}), \SA[e-1], v_{e-b}) = (\SA[e], v_{e-b+1})$ on $B(I_{\mathsf{SA}})$. 
The first sa-value $\SA[b]$ and the index $v_{0}$ are stored in the balanced sa-interval of $P$. 

The algorithm for a locate query also consists of two phases. 
In the first phase, 
the algorithm computes the balanced sa-interval of $P$ by iterating BSR query $m$ times. 
In the second phase, 
it computes $(e-b)$ move queries $\movef(B(I_{\mathsf{SA}}), \SA[b], v_{0})$, $\movef(B(I_{\mathsf{SA}}), \SA[b+1], v_{1})$, 
$\ldots$, $\movef(B(I_{\mathsf{SA}}), \SA[e-1], v_{e-b})$ by using the move data structure $F(I_{\mathsf{SA}})$, 
and outputs $\SA[b..e]$. 
Hence, we can solve a locate query in $O(m \log \log_{w} \sigma + occ)$ time.

%% file: 5_applications.tex
\section{Applications} \label{sec:applications}
In this section, we show that OptBWTR can support extract, decompression, and prefix search queries in optimal time.

\emph{Extract query.}
Let a string $T$ of length $n$ have $b$ marked positions $i_{1}, i_{2}, \ldots, i_{b} \in [1, n]$. 
An extract query (also called the bookmarking problem) is to return substring $T[i_{j}..i_{j}+d-1]$ for a given integer $j \in [1, b]$ and $d \in [1, n-i_{j}+1]$. 

We will use the \emph{FL function} to solve extract queries.
$\FL$ is the inverse function of LF function, 
i.e., $\FL(\LF(i)) = i$ holds for $i \in [1, n]$. 
We will also use the function $\FL_{x}$ and integers $h_{1}, h_{2}, \ldots, h_{b}$. 
$\FL_{x}(i)$ returns the position obtained by recursively applying the FL function to a given integer $i$ $x$ times, 
i.e., $\FL_{0}(i) = i$ and $\FL_{x}(i) = \FL_{x-1}(\FL(i))$ for $x \geq 1$. 
$h_{j}$ is the position with sa-value $i_{j}$ on SA~(i.e., $\SA[h_{j}] = i_{j}$). 
The FL function returns the position with the sa-value $x+1$ on SA for a given position with sa-value $x$, 
and hence $T[i_{j}..i_{j}+d-1] = F[\FL_{0}(h_{j})], F[\FL_{1}(h_{j})], \ldots, F[\FL_{d-1}(h_{j})]$ holds for $j \in [1, b]$, 
where $F$ is the string described in Section~\ref{sec:bwt}. 
We can construct a data structure of $O(r)$ words to compute FL function in constant time by modifying Theorem~\ref{theo:lf_inv_theorem} 
and can solve an extract query in linear time by using the data structure. 
See Appendix E.1 for details of our data structure for solving extract queries. 
\begin{theorem}\label{theo:bookmark}
There exists a data structure of $O(r + b)$ words that solves the bookmarking problem for a string $T$ and 
$b$ positions $i_{1}, i_{2}, \ldots, i_{b}$~($1 \leq i_{1} < i_{2} < \cdots < i_{b} \leq n$). 
We can construct this data structure in $O(n)$ time and $O(r + b)$ words of space by processing the RLBWT and positions $i_{1}, i_{2}, \ldots, i_{b}$. 
\end{theorem}
\begin{proof}
See Appendix E.1.
\end{proof}

\emph{Decompression of RLBWT.}
We apply Theorem~\ref{theo:bookmark} to $T[1..n]$ with marked position $1$. 
Then, our data structure for extract queries can return the string $T$ in $O(n)$ time (i.e., 
the data structure can recover $T$ from the RLBWT of $T$ in linear time to $n$).
The $O(n)$ time decompression is the fastest among other decompression algorithms on compressed indexes in $O(r)$ words of space, as the following theorem shows. 
\begin{theorem}
We can compute the characters of $T$ in left-to-right order~(i.e., $T[1], T[2]$, $\ldots$, $T[n]$) in $O(n)$ time and 
$O(r)$ words of space by processing the RLBWT of string $T$.
\end{theorem}

\emph{Prefix search.}
The prefix search for a set of strings $D = \{ T_{1}, T_{2}, \ldots, T_{d} \}$
returns the indexes of the strings in $D$ that include a given string $P$ as their prefixes~(i.e., 
$\{ i \mid i \in [1,d] \mbox{ s.t. } T_{i}[1..|P|] = P \}$). 
We can construct a data structure supporting the prefix search by combining Theorem~\ref{theo:bookmark} with \emph{compact trie}~\cite{DBLP:journals/jacm/Morrison68}.

A compact trie for a set of strings $D$ is a trie for $D$ such that all unary paths are collapsed, 
and each node represents the string by concatenating labels on the path from the root to the node. 
For simplicity, we will assume that any string in $D$ is not a prefix of any other string, 
and hence each leaf in the compact trie represents a distinct string in $D$. 
Let $v$ be the node such that (i) $P$ is a prefix of the string represented by the node 
and (ii) $P$ is not a prefix of the string represented by its parent. 
Then, the leaves under $v$ are the output of the prefix search query for $P$. 

To find $v$, we decode the string on the path from the root to the node $v$ in linear time using exact queries for the path. 
After we find $v$, we traverse the subtree rooted at $v$ and output all the leaves in the subtree. 
This procedure runs in $O(|P| + occ')$ time, where $occ'$ is the number of leaves under the lowest node. 
See Appendix E.2 for the details of our data structure for solving prefix search queries. 

\begin{theorem}\label{theo:prefix_search}
Let $r'$ be the number of runs in the RLBWT of a string $T$ containing all the strings in $D = \{ T_{1}, T_{2}, \ldots, T_{d} \}$. 
There exists a data structure that supports a prefix search on $D$ in $O(|P| + occ')$ time and $O(r' + d)$ words of space for a string $P$.
The data structure also returns the number of the strings in $D$ that include $P$ as their prefixes in $O(|P|)$ time.
\end{theorem}
\begin{proof}
See Appendix E.2.
\end{proof}

\section{Conclusion}
We presented OptBWTR, the first string index that can support count and locate queries on RLBWT in optimal time with $O(r)$ words of space for polylogarithmic alphabets.
OptBWTR also supports extract queries and prefix searches on RLBWT in optimal time for any alphabet size. 
In addition, we presented the first decompression algorithm working in optimal time and $O(r)$ words of working space. 
This is the first optimal-time decompression algorithm working in $O(r)$ words of space. 

We presented a new data structure of $O(r)$ words for computing LF and $\phi^{-1}$ functions 
in constant time by using a new data structure named move data structure, provided that we use an additional input. 
We also showed that the backward search works in optimal time for polylogarithmic alphabets with $O(r)$ words of space using the data structure. 
The two functions and the backward search are 
general and applicable to various queries on RLBWT. 

The following problems remain open:
Does there exist a string index of $O(r)$ words supporting locate queries in optimal time for any alphabet size? 
We assume $\sigma = O(\polylog n)$ for supporting locate queries in optimal time with $O(r)$ words. 
As mentioned in Section~\ref{sec:introduction}, 
a faster version of r-index can support locate queries in optimal time with $O(r \log \log_{w} (\sigma + (n/r)))$ words. 
Thus, improving OptBWTR so that it can support locate queries in optimal time with $O(r)$ words for any alphabet size is an important future work. 
For this goal, 
one needs to solve a rank query on a string of length $\Theta(r)$ in constant time and $O(r)$ words of space. 
However, this seems impossible because any data structure of $O(r)$ words requires $\Omega(\log\log_{w} \sigma)$ time 
to compute a rank query on a string of length $r$~\cite{DBLP:journals/talg/BelazzouguiN15}. 
Perhaps, we may be able to compute the sa-interval of a given pattern in $O(m)$ time and $O(r)$ words of space without using rank queries. 
After computing the sa-interval of the pattern, we can solve the locate query in optimal time by using our data structure for the $\phi^{-1}$ function.

%% file: 9_appendix.tex
\section*{Appendix A}

\subsection*{A.1: Proof of Lemma~\ref{lem:LF_property}}
\begin{proof}
(i) Let $y = (i-\ell_{x})$ and $c = L[\ell_{x} + (i-\ell_{x})]$.  
$\LF(\ell_{x} + y) = C[c] + \rank(L, c, \ell_{x} + y)$ holds 
by the BWT property described in Section~\ref{sec:bwt}. 
$\rank(L, c, \ell_{x} + y) = \rank(L, c, \ell_{x}) + y$ holds because the $x$-th run $L_{x}$ is a repetition of the character $c$. 
Hence $\LF(\ell_{x} + y) = C[c] + \rank(L, c, \ell_{x} + y) = C[c] + \rank(L, c, \ell_{x}) + y = \LF(\ell_{x}) + y$ holds. 
By $i = \ell_{x} + y$, 
$\LF(i) = \LF(\ell_{x}) + (i - \ell_{x})$ holds. 

(ii) 
Let $p$ be the integer satisfying $L_{p} = \$$. 
Then, $\LF(\ell_{p}) = 1$ holds, 
and hence $\LF(\ell_{\delta[1]}) = 1$ holds by $\delta[1] = p$. 
Next, $\LF(\ell_{\delta[i]}) = \LF(\ell_{\delta[i-1]}) + |L_{\delta[i-1]}|$ holds for any $i \in [2, r]$, 
because (a) 
the LF function maps the interval $[\ell_{\delta[i]}, \ell_{\delta[i]} + |L_{\delta[i]}| - 1]$ 
into the interval $[\LF(\ell_{\delta[i]}), \LF(\ell_{\delta[i]}) + |L_{\delta[i]}| - 1]$ by Lemma~\ref{lem:LF_property}(i) for any $i \in [1, r]$, 
(b) LF is a bijection from $[1, n]$ to $[1, n]$, 
and (c) $\LF(\ell_{\delta[1]}) < \LF(\ell_{\delta[2]}) < \cdots < \LF(\ell_{\delta[r]})$ holds. 

\end{proof}
\subsection*{A.2: Proof of Lemma~\ref{lem:phi_property}}
\begin{proof}
(i) 
Lemma~\ref{lem:phi_property}(i) clearly holds for $i = u_{x}$. 
We show that Lemma~\ref{lem:phi_property}(i) holds for $i \neq u_{x}$~(i.e., $i > u_{x}$). 
Let $s_{t}$ be the position with sa-value $u_{x} + t$ for an integer $t \in [1, y]$~(i.e., $\SA[s_{t}] = u_{x} + t$), 
where $y = i - u_{x}$.
Two adjacent positions $s_{t}$ and $s_{t}+1$ are contained in an interval $[\ell_{v}, \ell_{v} + |L_{v}| - 1]$ on SA~(i.e., 
$s_{t}, s_{t}+1 \in [\ell_{v}, \ell_{v} + |L_{v}| - 1]$), which corresponds to the $v$-th run $L_{v}$ of $L$. 
This is because $s_{t}$ is not the ending position of a run, i.e., 
$(u_{x} + t) \not \in \{ u_{1}, u_{2}, \ldots, u_{r} \}$. 
The LF function maps $s_{t}$ into $s_{t-1}$, 
where $s_{0}$ is the position with sa-value $u_{x}$. 
LF also maps $s_{t}+1$ into $s_{t-1}+1$ by Lemma~\ref{lem:LF_property}(i). 
The two mapping relationships established by LF produce
$y$ equalities $\phi^{-1}(\SA[s_{1}]) = \phi^{-1}(\SA[s_{0}]) + 1$, $\phi^{-1}(\SA[s_{2}]) = \phi^{-1}(\SA[s_{1}]) + 1$, 
$\ldots$, $\phi^{-1}(\SA[s_{y}]) = \phi^{-1}(\SA[s_{y-1}]) + 1$. 
The equalities lead to $\phi^{-1}(\SA[s_{y}]) = \phi^{-1}(\SA[s_{0}]) + y$, 
which represents $\phi^{-1}(i) = \phi^{-1}(u_{x}) + (i - u_{x})$ by 
$\SA[s_{y}] = i$, $\SA[s_{0}] = u_{x}$, and $y = i - u_{x}$. 

(ii) 
Let $p$ be the integer satisfying $L_{p} = \$$. 
Then there exists an integer $q$ such that $u_{q}$ is the sa-value at position $\ell_{p} - 1$ if $p \neq 1$; 
otherwise $u_{q}$ is the sa-value at position $n$. 
$\phi^{-1}(u_{q}) = 1$, because $\SA[\ell_{p}] = 1$ always holds. 
Hence $\phi^{-1}(u_{\delta'[1]}) = 1$ holds by $\delta'[1] = q$. 

Next, $\phi^{-1}(u_{\delta'[i]}) = \phi^{-1}(u_{\delta'[i-1]}) + d$ holds for any $i \in [2, r]$, 
because (a) 
$\phi^{-1}$ maps the interval $[u_{\delta'[i]}, u_{\delta'[i]} + d - 1]$ 
into the interval $[\phi^{-1}(u_{\delta'[i]}), \phi^{-1}(u_{\delta'[i]}) + d - 1]$ by Lemma~\ref{lem:phi_property}(i) for any $i \in [1, r]$, 
(b) $\phi^{-1}$ is a bijection from $[1, n]$ to $[1, n]$, 
and (c) $\phi^{-1}(u_{\delta'[1]}) < \phi^{-1}(u_{\delta'[2]}) < \cdots < \phi^{-1}(u_{\delta'[r]})$ holds. 

(iii) Recall that $p$ is the integer satisfying $L_{p} = \$$. 
Then there exists an integer $q'$ such that $u_{q'}$ is the sa-value at position $\ell_{p}$, 
because the length of $L_{p}$ is $1$. 
Hence, $u_{1} = u_{q'} = 1$ holds.

\end{proof}

\subsection*{A.3: Examples of Lemmas~\ref{lem:LF_property} and~\ref{lem:phi_property}}
Here, we give an example of Lemma~\ref{lem:LF_property}.
In Figure~\ref{fig:bwt}, 
$(\ell_{1}, \ell_{2}, \ell_{3}, \ell_{4}) = (1, 7, 13, 14)$ and 
$(\LF(\ell_{1}), \LF(\ell_{2})$, $\LF(\ell_{3}), \LF(\ell_{4})) = (10, 2, 1, 8)$. 
Hence, $\LF(3) = \LF(\ell_{1}) + (3 - \ell_{1}) = 12$ and 
$\LF(8) = \LF(\ell_{2}) + (8 - \ell_{2}) = 3$ hold by Lemma~\ref{lem:LF_property}(i).

Next, we give an example of Lemma~\ref{lem:phi_property}. 
In Figure~\ref{fig:bwt}, 
$(u_{1}, u_{2}, u_{3}, u_{4}) = (1, 4, 5, 9)$ and 
$(\phi^{-1}(u_{1}), \phi^{-1}(u_{2})$, $\phi^{-1}(u_{3})$, $\phi^{-1}(u_{4})) = (12, 15, 8, 1)$. 
Hence 
$\phi^{-1}(3) = \phi^{-1}(u_{1}) + (3 - u_{1}) = 14$ and 
$\phi^{-1}(8) = \phi^{-1}(u_{3}) + (8 - u_{3}) = 11$ hold by Lemma~\ref{lem:phi_property}(i).

\section*{Appendix B}
\subsection*{B.1: Proof of Lemma~\ref{lem:balance_size} }
\begin{proof}
Let $\mathcal{Q}(I^{\mathsf{out}}_{t-1})$ be the set of the starting positions of input intervals in $G(I^{\mathsf{out}}_{t-1})$~(i.e., 
$\mathcal{Q}(I^{\mathsf{out}}_{t-1}) = \{ p'_{1}, p'_{2}, \ldots, p'_{k'} \}$). 
Then $\mathcal{Q}(I^{\mathsf{out}}_{t}) = \mathcal{Q}(I^{\mathsf{out}}_{t-1}) \cup \{ p'_{j} + d \}$ holds from the definition of $I^{\mathsf{out}}_{t}$. 
Next, let $\mathsf{Edge}_{2}(I^{\mathsf{out}}_{t-1})$ be the set of output intervals such that 
each output interval has at least two incoming edges in $G(I^{\mathsf{out}}_{t-1})$, 
i.e., $\mathsf{Edge}_{2}(I^{\mathsf{out}}_{t-1}) = \{ [q'_{i}, q'_{i} + d'_{i} -1] \mid i \in [1, k'] \mbox{ s.t. } |[q'_{i}, q'_{i} + d'_{i} -1] \cap \mathcal{Q}(I^{\mathsf{out}}_{t-1})| \geq 2 \}$, 
where $d'_{i} = p'_{i+1} - p'_{i}$. 
$[q'_{i}, q'_{i} + d'_{i} -1] \in \mathsf{Edge}_{2}(I^{\mathsf{out}}_{t})$ holds if 
$[q'_{i}, q'_{i} + d'_{i} - 1] \in \mathsf{Edge}_{2}(I^{\mathsf{out}}_{t-1})$ for any integer $i \in [1, k'] \setminus \{ j \}$.
This is because (i) $[q'_{i}, q'_{i} + d'_{i} -1]$ is also an output interval of $I^{\mathsf{out}}_{t}$, 
and (ii) $([q'_{i}, q'_{i} + d'_{i} -1] \cap \mathcal{Q}(I^{\mathsf{out}}_{t-1})) \subseteq ([q'_{i}, q'_{i} + d'_{i} -1] \cap \mathcal{Q}(I^{\mathsf{out}}_{t}))$ holds by $\mathcal{Q}(I^{\mathsf{out}}_{t-1}) \subseteq \mathcal{Q}(I^{\mathsf{out}}_{t})$. 
$[q'_{j}, q'_{j} + d -1], [q'_{j} + d, q'_{j+1} - 1] \in \mathsf{Edge}_{2}(I^{\mathsf{out}}_{t})$ also holds by the third property of $I^{\mathsf{out}}_{t}$. 
Hence, we obtain an inequality $|\mathsf{Edge}_{2}(I^{\mathsf{out}}_{t})| \geq |\mathsf{Edge}_{2}(I^{\mathsf{out}}_{t-1})| + 1$ for any integer $t \in [1, \tau]$.  
The inequality $|\mathsf{Edge}_{2}(I^{\mathsf{out}}_{t})| \geq |\mathsf{Edge}_{2}(I^{\mathsf{out}}_{t-1})| + 1$ 
guarantees that $|\mathsf{Edge}_{2}(I^{\mathsf{out}}_{t})| \geq t$ holds for any integer $t \in [0, \tau]$. 
The inequality $|\mathsf{Edge}_{2}(I^{\mathsf{out}}_{t})| \geq t$ indicates that 
$I^{\mathsf{out}}_{t}$ consists of at least $2t$ pairs, 
because each output interval in $\mathsf{Edge}_{2}(I^{\mathsf{out}}_{t})$ has at least two incoming edges from distinct input intervals. 
Hence, Lemma~\ref{lem:balance_size} holds. 
\end{proof}

\subsection*{B.2: Details of $I_{\mathsf{SA}}$}
$I_{\mathsf{SA}}$ has the following three properties: 
(i) $u_{1} = 1 < u_{2} < \cdots < u_{r} \leq n$ holds by Lemma~\ref{lem:phi_property}(iii), 
(ii) $\phi^{-1}(u_{\delta'[1]}) = 1$ by Lemma~\ref{lem:phi_property}(ii), 
and (iii) 
$\phi^{-1}(u_{\delta'[i]}) = \phi^{-1}(u_{\delta'[i-1]}) + (u_{\delta'[i-1]+1} - u_{\delta'[i-1]})$ holds by Lemma~\ref{lem:phi_property}(ii) for each $i \in [2, r]$, 
where $\delta'$ is the permutation of $[1,r]$ introduced in Section~\ref{sec:bwt}. 
Hence $I_{\mathsf{SA}}$ satisfies the three conditions of the disjoint interval sequence, which were described in Section~\ref{sec:move_function}. 

Let $f_{\mathsf{SA}}$ be the bijective function represented by the disjoint interval sequence $I_{\mathsf{SA}}$. 
Then $f_{\mathsf{SA}}(i) = \phi^{-1}(u_{x}) + (i - u_{x})$ holds, 
where $x$ is the integer such that $u_{x} \leq i < u_{x+1}$ holds. 
On the other hand, $\phi^{-1}(i) = \phi^{-1}(u_{x}) + (i - u_{x})$ holds by Lemma~\ref{lem:phi_property}(i). 
Hence $f_{\mathsf{SA}}$ and $\phi^{-1}(i)$ are the same function.

\section*{Appendix C}
\subsection*{C.1: Rank-select data structure}
Here, we describe the set $\{ c_{1}, c_{2}, \ldots, c_{\sigma'} \}$ and function $\gamma$ for a string $S$.
$c_{1}, c_{2}, \ldots, c_{\sigma'}$ are all the distinct characters in $S$, 
i.e., $\{ c_{1}, c_{2}, \ldots, c_{\sigma'} \} = \{ S[i] \mid i \in [1, |S|] \}$~($c_{1} < c_{2} < \cdots < c_{\sigma'}$). 
The function $\gamma$ returns the rank of a given character $c \in \Sigma$ in a string $S$; i.e., 
$\gamma(S, c) = j$ if there exists an integer $j$ such that $c = c_{j}$ holds; otherwise $\gamma(S, c) = -1$.

A rank-select data structure $R(S)$ consists of three data structures $R_{\mathsf{rank}}$, $R_{\mathsf{select}}$, and  $R_{\mathsf{map}}$. 
$R_{\mathsf{rank}}$ is a \emph{rank data structure} for solving 
a rank query on a string $S$ in $O(\log\log_{w} \sigma)$ time and with $O(|S|)$ words of space~\cite{DBLP:journals/talg/BelazzouguiN15}. 
$R_{\mathsf{select}}$ consists of $\sigma'$ arrays $H_{1}, H_{2}, \ldots, H_{\sigma'}$. 
The size of $H_{j}$ is $|\Occ(S, c_{j})|$ for each $j \in \{ 1, 2, \ldots, \sigma' \}$, 
and $H_{j}[i]$ stores $\select(S, c_{j}, i)$ for each $i \in [1, |\Occ(S, c_{j})|]$. 
$R_{\mathsf{map}}$ is a \emph{deterministic dictionary}~\cite{DBLP:conf/icalp/Ruzic08} storing the mapping function $\gamma$ for $S$. 
The deterministic dictionary can compute $\gamma(S, c)$ for a given character $c$ in constant time, 
and its space usage is $O(\sigma')$ words. 
The space usage of the rank-select data structure is $O(|S|)$ words in total, 
because $\sigma' \leq |S|$ holds.
We can compute a given select query $\select(S, c, i)$ in two steps: 
(i) compute $j = \gamma(S, c)$; and  
(ii) return $-1$ if $j = -1$ or $|H_{j}| < i$; otherwise, return $H_{j}[i]$. 
Hence, the rank-select data structure can support rank and select queries on $S$ in $O(\log\log_{w} \sigma)$ and $O(1)$ time, respectively. 

\section*{Appendix D}
\subsection*{D.1: Proof of the construction time and working space in Theorem~\ref{theo:OptBWTR}}
Recall that OptBWTR consists of five data structures, i.e., $F(I_{\mathsf{LF}})$, $F(I_{\mathsf{SA}})$, $R(\sbwt)$, $\SA^{+}$, and $\SA^{+}_{\mathsf{index}}$.
Here, we will show that we can construct the five data structures in $O(n + r \log r)$ time and $O(r)$ words of working space 
by processing the RLBWT of $T$. 

\subsubsection*{Construction of $F(I_{\mathsf{LF}})$}
We use the following four lemmas. 
\begin{lemma}\label{lem:sort}
Let $K$ be a sequence of $k$ integers $p_{1}, p_{2}, \ldots, p_{k}$ such that $p_{i} = n^{O(1)}$ holds for each $i \in [1, k]$, 
where $k = O(n)$. 
We can sort the sequence $K$ in increasing order of those integers in $O(n)$ time and $O(k)$ words of working space. 
\end{lemma}
\begin{proof}
We sort the sequence $K$ by using a standard sorting algorithm if $k < \frac{n}{\log n}$. 
The sorting algorithm takes $O(k \log k) = O(n)$ time and $O(k)$ word of working space. 
Otherwise, we use an LSD radix sort with a bucket size of $\frac{n}{\log n}$; i.e., 
we sort the integers of $K$ with a bucket sort in $O(\frac{n}{\log n} + k)$ time by using the $\log (n / \log n)$ bits 
starting at position $(i-1)\log (n / \log n)$ 
for each step $i = 1, 2, \ldots, \frac{\log M}{\log (n / \log n)}$, 
where $M = \max\{ p_{1}, p_{2}, \ldots, p_{k} \}$. 
The space usage of the radix sort is $O(\frac{n}{\log n} + k) = O(k)$ words, and the running time is $O(n)$ time in total,  
because $\frac{\log M}{\log (n / \log n)} = O(1)$ holds by $M = O(n^{O(1)})$. 
Hence, we obtain Lemma~\ref{lem:sort}. 
\end{proof}

\begin{lemma}\label{lem:balanced_sequence_construction}
We can construct $B(I)$ in $O(k \log k)$ time and $O(k)$ words 
of working space for a given disjoint interval sequence $I$ of size $k$.
\end{lemma}
\begin{proof}
We define three self-adjusting balanced binary trees $\mathcal{T}^{\mathsf{in}}_{t}$, $\mathcal{T}^{\mathsf{out}}_{t}$, and $\mathcal{T}^{e}_{t}$ for $I^{\mathsf{out}}_{t}$, 
where $t$ is an integer in $[0, \tau]$.
The balanced binary tree $\mathcal{T}^{\mathsf{in}}_{t}$~(respectively, $\mathcal{T}^{\mathsf{out}}_{t}$) stores the pairs of $I^{\mathsf{out}}_{t}$ 
in increasing order of starting position of their input intervals~(respectively, their output intervals). 
The balanced binary tree $\mathcal{T}^{e}_{t}$ stores the pairs of $I^{\mathsf{out}}_{t}$ 
such that the output interval of each pair has at least four incoming edges in $G(I^{\mathsf{out}}_{t})$. 
In $\mathcal{T}^{e}_{t}$, the pairs of $I^{\mathsf{out}}_{t}$ are sorted in increasing order of starting position of their input intervals. 
$\mathcal{T}^{e}_{t}$ is empty if and only if $t = \tau$ for any $t \in [0, \tau]$. 

Now let us explain an algorithm 
for changing $\mathcal{T}^{\mathsf{in}}_{t-1}$ into $\mathcal{T}^{\mathsf{in}}_{t}$ 
and $\mathcal{T}^{\mathsf{out}}_{t-1}$ into $\mathcal{T}^{\mathsf{out}}_{t}$ by using $\mathcal{T}^{e}_{t-1}$.
Let $I^{\mathsf{out}}_{t-1} = (p'_{1}, q'_{1}), (p'_{2}, q'_{2}), \ldots, (p'_{k'}, q'_{k'})$ 
and $I^{\mathsf{out}}_{t} = (p'_{1}, q'_{1})$, $(p'_{2}, q'_{2})$, $\ldots$, $(p'_{j-1}, q'_{j-1})$, 
$(p'_{j}, q'_{j})$, $(p'_{j}+d, q'_{j}+d)$, $\ldots$, $(p'_{k'}, q'_{k'})$. 
After finding the $j$-th pair and integer $d$, 
we can update $\mathcal{T}^{\mathsf{in}}_{t-1}$ and $\mathcal{T}^{\mathsf{out}}_{t-1}$ in $O(\log k')$ time. 
The $j$-th pair $(p'_{j}, q'_{j})$ is the first pair stored in $\mathcal{T}^{e}_{t}$, 
and $d = p'_{x+1} - q'_{j} + 1$ holds. 
Here, $x$ is the smallest integer satisfying $p'_{x} \geq q'_{j}$. 
We can find $p'_{x+1}$ in $O(\log k')$ time by performing a binary search on $\mathcal{T}^{\mathsf{in}}_{t-1}$, 
and hence 
we can change $\mathcal{T}^{\mathsf{in}}_{t-1}$ and $\mathcal{T}^{\mathsf{out}}_{t-1}$ into 
$\mathcal{T}^{\mathsf{in}}_{t}$ and $\mathcal{T}^{\mathsf{out}}_{t}$, respectively, in $O(\log k')$ time. 

Next, let us explain an algorithm for changing $\mathcal{T}^{e}_{t-1}$ into $\mathcal{T}^{e}_{t}$ by using 
$\mathcal{T}^{\mathsf{in}}_{t}$, $\mathcal{T}^{\mathsf{out}}_{t}$, $j$, and $d$.
Let $y$ be the integer such that the $y$-th output interval $[q'_{y}, q'_{y} + d'_{y}-1]$ contains the position $p'_{j} + d$, 
and let $d'_{j} = p'_{j+1} - p'_{j}$.  
Two new input intervals $[p'_{j}, p'_{j} + d - 1]$ and $[p'_{j} + d, p'_{j} + d'_{j} - 1]$ are created by changing $I^{\mathsf{out}}_{t-1}$ into $I^{\mathsf{out}}_{t}$. 
In other words, the new input intervals and their corresponding output intervals~(i.e., $[q'_{j}, q'_{j} + d - 1], [q'_{j} + d, q'_{j} + d'_{j} - 1]$) 
are added to $G(I^{\mathsf{out}}_{t-1})$. 
The number of incoming edges of output interval $[q'_{y}, q'_{y} + d'_{y}-1]$ is also changed. 
As a result of the update, three pairs $(p'_{j}, q'_{j})$, $(p'_{j} + d, q'_{j} + d)$, and $(p'_{y}, q'_{y})$ may be added to $\mathcal{T}^{e}_{t-1}$ 
for changing $\mathcal{T}^{e}_{t-1}$ into $\mathcal{T}^{e}_{t}$.
On the other hand, the $j$-th input and output interval are removed from $G(I^{\mathsf{out}}_{t-1})$. 
As a result of the update, the $j$-th pair $(p'_{j}, q'_{j})$ is removed from $\mathcal{T}^{e}_{t-1}$.

We can change $\mathcal{T}^{e}_{t-1}$ into $\mathcal{T}^{e}_{t}$ in three steps:
(i) remove the $j$-th pair $(p'_{j}, q'_{j})$ from $\mathcal{T}^{e}_{t-1}$; 
(ii) find the pair creating the output interval $[q'_{y}, q'_{y+1}-1]$ by performing a binary search on $\mathcal{T}^{\mathsf{out}}_{t}$; 
(iii) verify whether each output interval $[b,e] \in \{ [q'_{j}, q'_{j} + d - 1], [q'_{j} + d, q'_{j} + d'_{j} - 1], [q'_{y}, q'_{y+1}-1] \}$ 
has at least four incoming edges in $G(I^{\mathsf{out}}_{t})$, 
and add the pair creating the output interval $[b,e]$ to $\mathcal{T}^{e}_{t-1}$ if $[b, e]$ has at least four incoming edges in $G(I^{\mathsf{out}}_{t})$. 
We can verify whether an output interval contains at least four incoming edges in $G(I^{\mathsf{out}}_{t})$ by using a binary search on $\mathcal{T}^{\mathsf{in}}_{t}$. 
Hence, the algorithm runs in $O(\log k')$ time in total. 

Next, we present an algorithm for constructing $I^{\mathsf{out}}_{\tau}$.
The construction algorithm consists of three phases: 
(i) construct $\mathcal{T}^{\mathsf{in}}_{0}$, $\mathcal{T}^{\mathsf{out}}_{0}$, and $\mathcal{T}^{e}_{0}$ by processing $I$ in $O(k \log k)$ time; 
(ii) construct $\mathcal{T}^{\mathsf{in}}_{\tau}$, $\mathcal{T}^{\mathsf{out}}_{\tau}$, and $\mathcal{T}^{e}_{\tau}$ 
using $\mathcal{T}^{\mathsf{in}}_{0}$, $\mathcal{T}^{\mathsf{out}}_{0}$, and $\mathcal{T}^{e}_{0}$;
and (iii) compute $I^{\mathsf{out}}_{\tau}$ using $\mathcal{T}^{\mathsf{in}}_{\tau}$. 
The second phase consists of $\tau$ steps. 
In the $t$-th step, we change $\mathcal{T}^{\mathsf{in}}_{t-1}$, $\mathcal{T}^{\mathsf{out}}_{t-1}$, and $\mathcal{T}^{e}_{t-1}$ 
into $\mathcal{T}^{\mathsf{in}}_{t}$, $\mathcal{T}^{\mathsf{out}}_{t}$, and $\mathcal{T}^{e}_{t}$, respectively, for $t \in [1, \tau]$. 
Hence, the construction algorithm of $I^{\mathsf{out}}_{\tau}$ takes $O(k \log k)$ time in total by $\tau \leq k$.
\end{proof}
\begin{lemma}\label{lem:balance_construction}
We can construct $F(I)$ in $O(k \log k)$ time and $O(k)$ words of working space 
for a given disjoint interval sequence $I$ of size $k$.
\end{lemma}
\begin{proof}
Recall that $F(I)$ consists of two arrays $D_{\mathsf{pair}}$ and $D_{\mathsf{index}}$. 
We construct $F(I)$ in two steps: 
(i) compute $B(I)$ by Lemma~\ref{lem:balanced_sequence_construction}; 
(ii) construct $D_{\mathsf{pair}}$ and $D_{\mathsf{index}}$ by processing $B(I)$. 
The construction takes $O(k \log k)$ time and $O(k)$ words of working space. 
\end{proof}

\begin{lemma}\label{lem:construction_I_LF}
We can compute $I_{\mathsf{LF}}$ in $O(n)$ time and $O(r)$ words of working space for a given RLBWT of $T$.
\end{lemma}
\begin{proof}
Recall that $I_{\mathsf{LF}}$ consists of $r$ pairs $(\ell_{1}, \LF(\ell_{1})), (\ell_{2}, \LF(\ell_{2})), \ldots, (\ell_{r}, \LF(\ell_{r}))$. 
$\LF(\ell_{\delta[1]}) = 1$, and 
$\LF(\ell_{\delta[i]}) = \LF(\ell_{\delta[i-1]}) + |L_{\delta[i-1]}|$ holds by Lemma~\ref{lem:LF_property}(ii) for any $i \in [2, r]$, 
where $\delta$ is the permutation of $[1,r]$ described in Section~\ref{sec:bwt}. 

We define $r$ integers $b_{1}, b_{2}, \ldots, b_{r}$. 
$b_{i}$ consists of $\log \sigma + \log r$ bits for each $i \in [1, r]$.  
The upper $\log \sigma$ bits and lower $\log r$ bits of $b_{i}$ represent the character $L[\ell_{i}]$ and $i$, respectively. 
$b_{\delta[1]} < b_{\delta[2]} < \cdots < b_{\delta[r]}$ holds, 
because 
$\LF(\ell_{i}) < \LF(\ell_{j})$ if and only if either of the following conditions holds: 
(i) $L[\ell_{i}] \prec L[\ell_{j}]$ or (ii) $L[\ell_{i}] = L[\ell_{j}]$ and $i < j$ for any pair of integers $i, j \in [1, r]$. 

We construct $I_{\mathsf{LF}}$ in five steps: 
(i) compute $\ell_{1}, \ell_{2}, \ldots, \ell_{r}$ and $b_{1}, b_{2}, \ldots, b_{r}$ by processing the RLBWT of $T$;
(ii) sort $b_{1}, b_{2}, \ldots, b_{r}$ by Lemma~\ref{lem:sort}; 
(iii) construct the permutation $\delta$ using the sorted $r$ integers $b_{1}, b_{2}, \ldots, b_{r}$; 
(iv) compute $\LF(\ell_{\delta[1]}), \LF(\ell_{\delta[2]}), \ldots, \LF(\ell_{\delta[r]})$ using 
$\ell_{1}, \ell_{2}, \ldots, \ell_{r}$ and the permutation $\delta$; 
and (v) construct $I_{\mathsf{LF}}$ using $\ell_{1}, \ell_{2}, \ldots, \ell_{r}$ and $\LF(\ell_{1}), \LF(\ell_{2}), \ldots, \LF(\ell_{r})$.
The construction time is $O(n)$ in total. 
\end{proof}

We construct $F(I_{\mathsf{LF}})$ in two steps: 
(i) construct $I_{\mathsf{LF}}$ by Lemma~\ref{lem:construction_I_LF}, 
and (ii) construct $F(I_{\mathsf{LF}})$ by Lemma~\ref{lem:balance_construction}. 
Hence we can construct $F(I_{\mathsf{LF}})$ in $O(n + r \log r)$ time and $O(r)$ words of working space 
by processing the RLBWT of $T$. 

\subsubsection*{Construction of $F(I_{\mathsf{SA}})$}
We will need the following lemma. 
\begin{lemma}\label{lem:construction_I_SA}
We can construct $I_{\mathsf{SA}}$ in $O(n + r \log r)$ and $O(r)$ words of working space using $F(I_{\mathsf{LF}})$ time for a given RLBWT of $T$.
\end{lemma}
\begin{proof}
Let $\LF_{x}(i)$ be the position obtained by recursively applying the LF function to $i$ $x$ times, 
i.e., $\LF_{0}(i) = i$ and $\LF_{x}(i) = \LF_{x-1}(\LF(i))$ for any integer $x \geq 1$. 
$\SA[\LF_{x}(p)] = n - x$ holds for any integer $x \in [0, n-1]$, 
where $p$ is the position on SA such that $\SA[p] = n$~(i.e., $p=1$). 
This is because $\SA[\LF(i)] = \SA[i] - 1$ holds unless $\SA[i] \neq 1$. 

Next, recall that $I_{\mathsf{SA}}$ consists of $r$ pairs $(u_{1}, \phi^{-1}(u_{1}))$, $(u_{2}, \phi^{-1}(u_{2}))$, $\ldots$, $(u_{r}$, $\phi^{-1}(u_{r}))$, 
where $u_{1}, u_{2}, \ldots, u_{r}$ are the integers described in Section~\ref{sec:bwt}. 
Let $x(1), x(2), \ldots, x(r)$ be distinct integers in $[0, n-1]$ such that 
$\LF_{x(i)}(p) \in \mathcal{X}$ holds for each $i \in [1, r]$, 
where $x(1) > x(2) > \cdots > x(r)$. 
Here, $\mathcal{X}$ is the set of the positions corresponding to the last characters of runs $L[1], L[2], \ldots, L[r]$, 
i.e., $\mathcal{X} = \{ \ell_{1} + |L_{1}| - 1, \ell_{2} + |L_{2}| - 1, \ldots, \ell_{r} + |L_{r}| - 1 \}$.
$u_{i} = n - x(i)$ holds, 
because (i) $n - x(i)$ is the sa-value at position $\LF_{x(i)}(p)$, 
(ii) $\LF_{x(i)}(p)$ is a position corresponding to the last character of a run, 
and (iii) $u_{1}, u_{2}, \ldots, u_{r}$ are the sa-values on the last characters of runs $L[1], L[2], \ldots, L[r]$. 

Next, let $x'(1), x'(2), \ldots, x'(r)$ be distinct integers in $[0, n-1]$ such that 
$\LF_{x'(i)}(p) \in \mathcal{X}'$ holds for each $i \in [1, r]$, 
where $x'(1) > x'(2) > \cdots > x'(r)$. 
Here, $\mathcal{X}'$ is the set of the positions corresponding to the first characters of runs $L[1], L[2], \ldots, L[r]$, 
i.e., $\mathcal{X}' = \{ \ell_{1}, \ell_{2}, \ldots, \ell_{r} \}$.
$\phi^{-1}(u_{i}) = n - x'(j)$ holds for each $i \in [1, r]$, 
where $j$ is an integer such that $\LF_{x'(j)}(p) = \LF_{x(i)}(p) + 1$ if $\LF_{x(i)}(p) \neq n$ 
and $\LF_{x'(j)}(p) = 1$ otherwise. 
This is because (i) $\phi^{-1}$ returns the sa-value at position $x+1$ for the sa-value at a position $x$, 
(ii) $u_{1}, u_{2}, \ldots, u_{r}$ are the sa-values on the last characters of runs $L[1], L[2], \ldots, L[r]$, 
and (iii) $\phi^{-1}(u_{1}), \phi^{-1}(u_{2}), \ldots, \phi^{-1}(u_{r})$ are the sa-values at the first characters of runs $L[1], L[2], \ldots, L[r]$.
Hence, we can compute 
$u_{1}, u_{2}, \ldots, u_{r}$ and $\phi^{-1}(u_{1}), \phi^{-1}(u_{2}), \ldots, \phi^{-1}(u_{r})$ in $O(r \log r)$ time 
using a standard sorting algorithm 
after constructing two sets $U = \{ (x(1), \LF_{x(1)}(p))$, $(x(2), \LF_{x(2)}(p))$, $\ldots$, $(x(r), \LF_{x(r)}(p)) \}$ and 
$U' = \{ (x'(1), \LF_{x'(1)}(p)),$ $(x'(2), \LF_{x'(2)}(p))$, $\ldots$, $(x'(r)$, $\LF_{x'(r)}(p)) \}$. 
These two sets can be computed by using $n$ move queries on $B(I_{\mathsf{LF}})$. 

Our construction algorithm of $I_{\mathsf{SA}}$ consists of three steps: 
(i) construct two sets $U$ and $U'$ using $n$ move queries on $B(I_{\mathsf{LF}})$; 
(ii) compute $u_{1}, u_{2}, \ldots, u_{r}$ and $\phi^{-1}(u_{1})$, $\phi^{-1}(u_{2})$, $\ldots$, $\phi^{-1}(u_{r})$ in $O(r \log r)$ time 
using the two sets $U$ and $U'$; 
(iv) construct $I_{\mathsf{SA}}$ using $u_{1}, u_{2}, \ldots, u_{r}$ and $\phi^{-1}(u_{1})$, $\phi^{-1}(u_{2})$, $\ldots$, $\phi^{-1}(u_{r})$.
The construction time is $O(n + r \log r)$ in total.
\end{proof}

We construct $F(I_{\mathsf{LF}})$ in two steps: 
(i) construct $I_{\mathsf{SA}}$ by Lemma~\ref{lem:construction_I_SA}, 
and (ii) construct $F(I_{\mathsf{SA}})$ by Lemma~\ref{lem:balance_construction}. 
Hence we can construct $F(I_{\mathsf{SA}})$ in $O(n + r \log r)$ time and $O(r)$ words of working space 
by processing the RLBWT of $T$. 

\subsubsection*{Construction of $R(\sbwt)$}
Let $B(I_{\mathsf{LF}}) = (p'_{1}, q'_{1}), (p'_{2}, q'_{2}), \ldots, (p'_{r'}, q'_{r'})$. 
We will need the following two lemmas. 
\begin{lemma}\label{lem:construction_sbwt}
We can compute $\sbwt$ in $O(r)$ time and $O(r)$ words of working space using $F(I_{\mathsf{LF}})$ for a given RLBWT of $T$.
\end{lemma}
\begin{proof}
Let $s(i)$ be the largest integer such that $\ell_{s(i)} \leq p'_{i}$ for each $i \in [1, r']$. 
Then $\sbwt[i] = L[\ell_{s(i)}]$ holds and 
$p'_{1}, p'_{2}, \ldots, p'_{r'}$ are stored in $F(I_{\mathsf{LF}})$. 
Our algorithm for constructing $\sbwt$ consists of two steps: 
(i) compute $s(1), s(2), \ldots, s(r')$ using $F(I_{\mathsf{LF}})$ and the RLBWT of $T$; 
and (ii) construct $\sbwt$ using $s(1), s(2), \ldots, s(r')$ and the RLBWT of $T$. 
The construction time is $O(r + r') = O(r)$ in total, because $r' \leq 2r$. 
\end{proof}

\begin{lemma}\label{lem:construction_rank_select_data_structure}
The rank-select data structure for $\sbwt$ can be constructed in $O(n + r (\log \log r)^{2})$ time and $O(r)$ words of working space by processing $\sbwt$. 
\end{lemma}
\begin{proof}
Recall that the rank-select data structure consists of $R_{\mathsf{rank}}$, $R_{\mathsf{select}}$, and $R_{\mathsf{map}}$. 
First, we will construct the deterministic dictionary $R_{\mathsf{map}}$ in two steps: 
(i) sort the characters of $\sbwt$ by Lemma~\ref{lem:sort}; 
and (ii) construct $R_{\mathsf{map}}$ from the sorted characters. 
The second step takes $O(r' (\log \log r')^{2})$ time and $O(r')$ word of working space~\cite{DBLP:conf/icalp/Ruzic08}, 
and hence we can construct $R_{\mathsf{map}}$ in $O(n + r' (\log \log r')^{2})$ time by processing $\sbwt$.

Second, we construct $R_{\mathsf{rank}}$ for $\sbwt$ by processing $\sbwt$. 
The construction takes $O(r')$ time and words of working space~\cite{DBLP:journals/talg/BelazzouguiN15}. 
Third, we construct $R_{\mathsf{select}}$. 
Recall that $R_{\mathsf{select}}$ consists of $\sigma'$ arrays $H_{1}, H_{2}, \ldots, H_{\sigma'}$. 
Our construction algorithm of $R_{\mathsf{select}}$ consists of $r'$ steps. 
We create $\sigma'$ empty arrays $H_{1}, H_{2}, \ldots, H_{\sigma'}$ before the construction algorithm runs. 
At the $i$-th step, we compute $c' = \gamma(\sbwt, \sbwt[i])$ using $R_{\mathsf{map}}$ 
and push integer $i$ into $H_{c'}$. 
The construction algorithm takes $O(r')$ time. 
Hence, we can construct $R_{\mathsf{rank}}$, $R_{\mathsf{select}}$, and 
$R_{\mathsf{map}}$ in $O(n + r' (\log \log r')^{2})$ time and $O(r')$ words of working space. 
Finally, Lemma~\ref{lem:construction_rank_select_data_structure} follows from $r' \leq 2r$. 
\end{proof}

We construct $R(\sbwt)$ in three steps: 
(i) construct $F(I_{\mathsf{LF}})$ by processing the RLBWT of $T$, 
(ii) construct $\sbwt$ by Lemma~\ref{lem:construction_sbwt}, 
and (ii) construct $R(\sbwt)$ by Lemma~\ref{lem:construction_rank_select_data_structure}. 
Hence, we can construct $R(\sbwt)$ in $O(n + r \log r)$ time and $O(r)$ words of working space 
by processing the RLBWT of $T$.

\subsubsection*{Construction of $\SA^{+}$ and $\SA^{+}_{\mathsf{index}}$}
We will need the following lemma. 
\begin{lemma}\label{lem:construction_sa_index}
We can construct $\SA^{+}$ and $\SA^{+}_{\mathsf{index}}$ in $O(n + r \log r)$ time and $O(r)$ words of working space 
using $F(I_{\mathsf{LF}})$ and $F(I_{\mathsf{SA}})$.
\end{lemma}
\begin{proof}
Let $B(I_{\mathsf{LF}}) = (p_{1}, q_{1}), (p_{2}, q_{2}), \ldots, (p_{k}, q_{k})$ and 
$p$ be the integer satisfying $\SA[p] = n$~(i.e., $p = 1$). 
We introduce $k$ integers $\mu(1), \mu(2), \ldots, \mu(k)$. 
Here, $\mu(i)$ is the integer satisfying $\LF_{\mu(i)}(p) = p_{i}$ for $i \in [1, k]$. 
$\SA[p_{i}] = n - \mu(i)$; hence, $\SA^{+}[i] = n - \mu(i)$.
The $k$ integers $p_{1}, p_{2}, \ldots, p_{k}$ are stored in $F(I_{\mathsf{LF}})$, 
and we can compute $\LF_{1}(p), \LF_{2}(p), \ldots, \LF_{n-1}(p)$ in left-to-right order using $F(I_{\mathsf{LF}})$. 
Hence $\SA^{+}$ can be constructed in $O(n)$ time. 

Next, 
let $B(I_{\mathsf{SA}}) = (p'_{1}, q'_{1}), (p'_{2}, q'_{2}), \ldots, (p'_{k'}, q'_{k'})$. 
Recall that $\SA^{+}_{\mathsf{index}}[i]$ stores the index $j$ of the input interval of $B(I_{\mathsf{SA}})$ containing $\SA^{+}[i]$ for each $i \in [1, k']$~(i.e., 
$j$ is the integer satisfying $p'_{j} \leq \SA^{+}[i] < p'_{j+1}$).
We find $\SA^{+}_{\mathsf{index}}[i]$ by performing a binary search on $p'_{1}, p'_{2}, \ldots, p'_{k'}$. 
The $k'$ integers $p'_{1}, p'_{2}, \ldots, p'_{k'}$ are stored in $F(I_{\mathsf{SA}})$; 
hence, we can construct $\SA^{+}_{\mathsf{index}}[i]$ in $O(r \log r)$ after constructing $\SA^{+}$. 
\end{proof}

We construct $\SA^{+}$ and $\SA^{+}_{\mathsf{index}}$ in two steps: 
(i) construct $F(I_{\mathsf{LF}})$ and $F(I_{\mathsf{SA}})$ by processing the RLBWT of $T$, 
and (ii) construct $\SA^{+}$ and $\SA^{+}_{\mathsf{index}}$ by Lemma~\ref{lem:construction_sa_index}. 
Hence we can construct $\SA^{+}$ and $\SA^{+}_{\mathsf{index}}$ in $O(n + r \log r)$ time and $O(r)$ words of working space 
by processing the RLBWT of $T$.

\section*{Appendix E}
\subsection*{E.1: Proof of Theorem~\ref{theo:bookmark}}
We will consider two cases for Theorem~\ref{theo:bookmark}: (a) $r < \frac{n}{\log n}$; (b) $r \geq \frac{n}{\log n}$. 
\subsection*{Case (a)}
\emph{Data structure}. 
In this case, we use a data structure $E_{\mathsf{first}}$ supporting extract queries. 
We will use five symbols $I_{\mathsf{FL}}$, $F'$, $\flbwt$, $V$, and $G$ to explain the data structure $E_{\mathsf{first}}$. 
Recall that $\ell_{i}$ is the starting position of the $i$-th run of BWT $L$, 
and $\delta$ is the permutation of $[1,r]$ described in Section~\ref{sec:bwt}. 
The first symbol $I_{\mathsf{FL}}$ is a sequence of $r$ pairs $(\LF(\ell_{\delta[1]}), \ell_{\delta[1]})$, $(\LF(\ell_{\delta[2]}), \ell_{\delta[2]})$, $\ldots$, $(\LF(\ell_{\delta[r]}), \ell_{\delta[r]})$. 
LF is a bijection and FL is the inverse of LF. 
Hence, $I_{\mathsf{FL}}$ is a disjoint interval sequence. 
Let $f_{\mathsf{FL}}$ be the bijective function represented by the disjoint interval sequence $I_{\mathsf{FL}}$. 
Then, $f_{\mathsf{FL}}$ and $\FL$ are the same function. 

Next, $F'$ is a sequence of strings $F[\LF(\ell_{\delta[1]})..\LF(\ell_{\delta[2]})-1]$, $F[\LF(\ell_{\delta[2]})..\LF(\ell_{\delta[3]})-1]$, $\ldots$, 
$F[\LF(\ell_{\delta[r]})..n]$. 
The FL function maps each string of $F'$ into the corresponding run of $L$, 
and hence, each string of $F'$ is a repetition of a character. 

The third symbol $\flbwt$ is the string satisfying $\flbwt = F[\LF(\ell_{\delta[1]})], F[\LF(\ell_{\delta[2]})], \ldots, F[\LF(\ell_{\delta[r]})]$. 
Let $v(i)$ be the index of the input interval of $I_{\mathsf{FL}}$ containing a position $i \in [1, n]$. 
The $v(i)$-th input interval corresponds to the $v(i)$-th string of $F'$, 
and the $v(i)$-th string of $F'$ is a repetition of the $v(i)$-th character of $\flbwt$. 
Hence $F[i] = \flbwt[v(i)]$ holds. 

Let $B(I_{\mathsf{FL}}) = (p'_{1}, q'_{1})$, $(p'_{2}, q'_{2})$, $\ldots$, $(p'_{k'}, q'_{k'})$ 
be the balanced interval sequence of $I_{\mathsf{FL}}$.  
The fourth symbol $V$ is an array of size $|B(I_{\mathsf{FL}})|$. 
$V[j]$ stores the index of the input interval of $I_{\mathsf{FL}}$ containing position $p'_{j}$ for each $j \in [1, |B(I_{\mathsf{FL}})|]$. 
Let $v'(i)$ be the index of the input interval of $B(I_{\mathsf{FL}})$ containing a position $i \in [1, n]$. 
Each $j$-th input interval of $B(I_{\mathsf{FL}})$ is contained in the $V[j]$-th input interval of $I_{\mathsf{FL}}$, 
and the input interval of $I_{\mathsf{FL}}$ corresponds to the $V[j]$-th string of $F'$, which is a repetition of the character $\flbwt[V[j]]$. 
Hence $F[i] = \flbwt[V[v'(i)]]$ holds. 

The fifth symbol $G$ is an array of size $b$. 
Recall that $h_{j}$ is the position with sa-value $i_{j}$ in SA~(i.e., $\SA[h_{j}] = i_{j}$). 
$G[j]$ stores a pair of integers $h_{j}$ and $v'(h_{j})$ for each $j \in [1, b]$. 
Our data structure $E_{\mathsf{first}}$ consists of $\flbwt$, $F(I_{\FL})$, $V$, and $G$. 
The space usage is $O(|B(I_{\mathsf{FL}})| + b) = O(r + b)$ words, because $|B(I_{\mathsf{FL}})| \leq 2r$. 

Now let us explain an algorithm for solving extract queries that uses $E_{\mathsf{first}}$ to return a substring $T[i_{j}..i_{j} + d - 1]$. 
Recall that $T[i_{j}..i_{j}+d-1] = F[\FL_{0}(h_{j})], F[\FL_{1}(h_{j})], \ldots, F[\FL_{d-1}(h_{j})]$ holds. 
On the other hand, 
$\movef(B(I_{\mathsf{FL}}), \FL_{0}(h_{j})$, $v'(\FL_{0}(h_{j}))) = (\FL_{1}(h_{j}), v'(\FL_{1}(h_{j})) )$, 
$\movef(B(I_{\mathsf{FL}})$, $\FL_{1}(h_{j})$, $v'(\FL_{1}(h_{j})) ) = (\FL_{2}(h_{j}), v'(\FL_{2}(h_{j})) )$, $\ldots$, 
$\movef(B(I_{\mathsf{FL}}), \FL_{d-2}(h_{j})$, $v'(\FL_{d-2}(h_{j})) ) = (\FL_{d-1}(h_{j}), v'(\FL_{d-1}(h_{j})) )$ 
hold. 
The first two integers $\FL_{0}(h_{j}) = h_{j}$ and $v'(\FL_{0}(h_{j})) = v'(h_{j})$ are stored in $G[j]$, 
and $F[\FL_{t}(h_{j})]$ is equal to $\flbwt[V[v'(\FL_{t}(h_{j})) ]]$. 

Our algorithm consists of $d$ steps. 
At the $t$-th step~($t \in [1, d]$), 
the algorithm returns character $F[\FL_{t-1}(h_{j})]$ as the $t$-th character of the output string $T[i_{j}..i_{j} + d - 1]$,  
and it executes the move query $\movef(B(I_{\mathsf{FL}}), \FL_{t-1}(h_{j}), v'(\FL_{t-1}(h_{j})) )$ for the next step. 
Hence, we can solve an extract query in constant time per character using $E_{\mathsf{first}}$. 

\emph{Construction of $E_{\mathsf{first}}$}. 
Next, we explain an algorithm for constructing $E_{\mathsf{first}}$. 
We will leverage the following two lemmas.
\begin{lemma}\label{lem:construction_FL}
We can construct $I_{\mathsf{FL}}$ and $\flbwt$ in $O(n)$ time and $O(r)$ words of working space by processing the RLBWT of $T$.
\end{lemma}
\begin{proof}
$I_{\mathsf{FL}}$ consists of $r$ pairs 
$(\LF(\ell_{\delta[1]}), \ell_{\delta[1]})$, $(\LF(\ell_{\delta[2]}), \ell_{\delta[2]})$, $\ldots$, $(\LF(\ell_{\delta[r]}), \ell_{\delta[r]})$, 
and $I_{\mathsf{LF}}$ consists of $r$ pairs 
$(\ell_{1}, \LF(\ell_{1}))$, $(\ell_{2}, \LF(\ell_{2}))$, $\ldots$, $(\ell_{r}, \LF(\ell_{r}))$. 
We construct $I_{\mathsf{FL}}$ in four steps: 
(i) construct $I_{\mathsf{LF}}$ by Lemma~\ref{lem:construction_I_LF}; 
(ii) sorting $\LF(\ell_{1})$, $\LF(\ell_{2})$, $\ldots$, $\LF(\ell_{r})$ by Lemma~\ref{lem:sort}; 
(iii) construct the permutation $\delta$ using the sorted $r$ integers $\LF(\ell_{1})$, $\LF(\ell_{2})$, $\ldots$, $\LF(\ell_{r})$; 
and (iv) construct $I_{\mathsf{FL}}$ using $I_{\mathsf{LF}}$ and the permutation $\delta$. 
The construction algorithm takes $O(n)$ time. 

Next, $\flbwt[i]$ is the first character of the $\delta[i]$-th run of BWT $L$, 
and the first characters of runs in BWT $L$ are stored in the RLBWT of $T$. 
Hence, we can construct $\flbwt$ in $O(r)$ time after constructing $\delta$.

\end{proof}

\begin{lemma}\label{lem:construction_G}
We can construct $G$ in $O(n)$ time and $O(r)$ words of working space using $F(I_{\mathsf{FL}})$ and $b$ marked positions $i_{1}, i_{2}, \ldots, i_{b}$.
\end{lemma}
\begin{proof}
Let $p$ be the integer satisfying $\SA[p] = n$~(i.e., $p = 1$). 
Then $G[j] = (\FL_{x}(p), v'(\FL_{x}(p)))$ holds for any $j \in [1, b]$, 
where $x = i_{j}$.
We can compute $n$ pairs $(\FL_{1}(p), v'(\FL_{1}(p)))$, $(\FL_{2}(p), v'(\FL_{2}(p)))$, $\ldots$, $(\FL_{n}(p), v'(\FL_{n}(p)))$ 
using $n$ move queries on $B(I_{\mathsf{FL}})$. 
Hence we can construct $G$ in $O(n)$ time and $O(r)$ words of working space. 

\end{proof}
We construct $E_{\mathsf{first}}$ in four steps. 
(i) compute $I_{\mathsf{FL}}$ and $\flbwt$ by Lemma~\ref{lem:construction_FL}; 
(ii) construct $F(I_{\mathsf{FL}})$ by Lemma~\ref{lem:balance_construction}; 
(iii) construct $V$ by processing $I_{\mathsf{FL}}$ and $B(I_{\mathsf{FL}})$; and 
(v) construct $G$ by Lemma~\ref{lem:construction_G}. 
Hence the construction time is $O(n + r\log r) = O(n)$ in total by $r < \frac{n}{\log n}$. 

\subsection*{Case (b)}
\emph{Data structure}. 
In this case, we use a data structure $E_{\mathsf{sec}}$ supporting extract queries. 
We will use three symbols $G'$, $S'$, and $R_{\mathsf{pred}}$ to explain $E_{\mathsf{sec}}$. 
$G'$ is an array of size $b$, 
and $G'[j]$ stores a pair of integers $h_{j}$ and $v(h_{j})$ for each $j \in [1, b]$. 
$S'$ is the set $\{ \LF(\ell_{\delta[1]}), \LF(\ell_{\delta[2]}), \ldots, \LF(\ell_{\delta[r]}) \}$. 

$R_{\mathsf{pred}}(S)$ is a data structure supporting predecessor queries on a given set $S \subseteq [1, n]$ of 
integers. It was proposed by Belazzougui and Navarro \cite{DBLP:journals/talg/BelazzouguiN15}. 
The data structure requires $O(|S|)$ words and supports a predecessor query in $O(\log \log_{w} (n/|S|))$ time. 
Our data structure $E_{\mathsf{sec}}$ consists of $\flbwt$, $I_{\FL}$, $R_{\mathsf{pred}}(S')$, and $G'$. 
The space usage is $O(r + b)$ words. 
At this point, we can prove the following lemma.  
\begin{lemma}\label{lem:predecessor}
We can compute a pair $(\FL(i), v(\FL(i)))$ in $O(\log \log_{w} (n/r))$ time 
by using $I_{\mathsf{FL}}$, $R_{\mathsf{pred}}(S')$, and $v(i)$ for a given integer $i \in [1, n]$. 
\end{lemma}
\begin{proof}
$v(j) = |\{ x \mid x \in S' \mbox{ s.t. } x \leq j \}|$ holds for any integer $j \in [1, n]$. 
Hence we can compute $v(\FL(i))$ using a predecessor query on $S$ after computing $\FL(i)$. 
Next, 
$\FL(i) =  \ell_{\delta[x]} + (i - \LF(\ell_{\delta[x]}))$ holds by Lemma~\ref{lem:LF_property}(i), 
where $x$ is the integer such that $\LF(\ell_{\delta[x]}) \leq i < \LF(\ell_{\delta[x+1]})$. 
This is because FL is the inverse of the LF function. 
$\FL(i) =  \ell_{\delta[v(i)]} + (i - \LF(\ell_{\delta[v(i)]}))$ also holds, 
because $x = v(i)$. 
$\ell_{\delta[v(i)]}$ and $\LF(\ell_{\delta[v(i)]})$ are stored in $I_{\mathsf{FL}}$; 
hence, we can compute $\FL(i)$ using $I_{\mathsf{FL}}$. 

We compute the pair $(\FL(i), v(\FL(i))$ in two steps: 
(i) compute $\FL(i)$ by $\FL(i) =  \ell_{\delta[v(i)]} + (i - \LF(\ell_{\delta[v(i)]}))$, 
and (ii) find $v(\FL(i))$ by a predecessor query on $S'$. 
The running time is $O(\log \log_{w} (n/r))$. 

\end{proof}

Now let us explain an algorithm for solving extract queries that uses $E_{\mathsf{sec}}$ to return a substring $T[i_{j}..i_{j} + d - 1]$.  
Recall that $T[i_{j}..i_{j}+d-1] = F[\FL_{0}(h_{j})], F[\FL_{1}(h_{j})], \ldots, F[\FL_{d-1}(h_{j})]$, 
and $F[\FL_{t}(h_{j})] = \flbwt[v(\FL_{t}(h_{j}))]$ for $t \in [0, d-1]$. 
The two integers $h_{j}$ and $v(h_{j})$ are stored in $G'[j]$. 
Our algorithm consists of $d$ steps. 
In the $t$-th step~($t \in [1, d]$), 
the algorithm returns the character $\flbwt[v(\FL_{t-1}(h_{j}))]$ as the $t$-th character of the output string $T[i_{j}..i_{j} + d - 1]$,  
and it computes pair $(\FL_{t}(h_{j}), v(\FL_{t}(h_{j})))$ by Lemma~\ref{lem:predecessor} for the next step. 
Hence, we can solve an extract query in $O(\log \log_{w} (n/r))$ per character using $E_{\mathsf{sec}}$, 
and $O(\log \log_{w} (n/r)) = O(1)$ because $r \geq \frac{n}{\log n}$.

\emph{Construction of $E_{\mathsf{sec}}$}. 
We construct $E_{\mathsf{sec}}$ in three steps: 
(i) compute $I_{\mathsf{FL}}$, $\flbwt$, and $S$ by Lemma~\ref{lem:construction_FL}; 
(ii) construct $R_{\mathsf{pred}}(S)$ by processing $S$;  
(iii) construct $G'$. 
Constructing $R_{\mathsf{pred}}(S)$ takes $O(|S| \log\log_{w} (n/|S|))$ time and $O(|S|)$ words~\cite{10.1145/3375890}. 
In the third step, we construct $G'$ by the following lemma. 
\begin{lemma}\label{lem:construction_G_dash}
We can construct $G'$ in $O(n \log \log_{w} (n/r))$ time and $O(r)$ words of working space using $R_{\mathsf{pred}}(S)$ and $b$ marked positions $i_{1}, i_{2}, \ldots, i_{b}$.
\end{lemma}
\begin{proof}
Let $p$ be the integer satisfying $\SA[p] = n$~(i.e., $p = 1$). 
Then $G[j] = (\FL_{x}(p), v(\FL_{x}(p)))$ holds for any $j \in [1, b]$, 
where $x = i_{j}$.
We can compute $n$ pairs $(\FL_{1}(p), v(\FL_{1}(p)))$, $(\FL_{2}(p), v(\FL_{2}(p)))$, $\ldots$, $(\FL_{n}(p), v(\FL_{n}(p)))$ 
in $O(n \log \log_{w} (n/r))$ time by Lemma~\ref{lem:predecessor}. 
Hence, we can construct $G'$ in $O(n \log \log_{w} (n/r))$ time and $O(r)$ words of working space. 

\end{proof}
Hence, the construction time is $O(n \log \log_{w} (n/r))$ in total, 
and $O(n\log \log_{w} (n/r)) = O(n)$ holds because $r \geq \frac{n}{\log n}$.

\subsection*{E.2: Proof of Theorem~\ref{theo:prefix_search}}
\begin{proof}
We use the following two data structures. 
The first data structure is a compact trie for $D$ without the strings on the edges where each internal node stores 
(i) the number of leaves under the node, (ii) the pointers to the leftmost and rightmost leaves under the node and 
(iii) a perfect hash table storing the first character of the outgoing edges; each leaf stores the index of the string represented by 
the leaf and the pointer to the right leaf. 
The second data structure is the data structure presented in Theorem~\ref{theo:bookmark}, and it stores the strings on the edges. 
The data structure can be stored in $O(r' + d)$ words.

Recall that $v$ is a node such that (i) $P$ is a prefix of the string represented by the node 
and (ii) $P$ is not a prefix of the string represented by its parent. 
To answer a prefix search query, we traverse the path from the root to the node $v$ and output the indexes stored in the leaves under $v$. 
We can execute this procedure in $O(|P| + occ')$ time. 
We can also compute the number of the leaves~(i.e., the number of strings in $D$ that include $P$ as their prefixes) in $O(|P|)$ time using the two data structures. 
\end{proof}